\documentclass[12pt,a4paper]{article}
\usepackage{fullpage}
\usepackage{authblk}
\usepackage{mymacros}
\usepackage{hyperref}
\usepackage{cleveref}
\usepackage{crossreftools}
\usepackage{setspace}

\crefname{proofcntr}{proof}{proofs}
\Crefname{proofcntr}{Proof}{Proofs}

\usepackage{comment}
\usepackage{bm}
\usepackage{verbatim}
\usepackage{mathtools}
\DeclarePairedDelimiter\ceil{\lceil}{\rceil}
\DeclarePairedDelimiter\floor{\lfloor}{\rfloor}

\newcommand{\good}{\textsc{good}}
\newcommand{\bad}{\textsc{bad}}
\newcommand{\tsigma}{\Tilde{\sigma}}

\newcommand{\opt}{\textsc{OPT}}
\newcommand{\wP}{\widetilde{P}}
\newcommand{\hmax}{H_{\max}}

\newcommand{\cB}{\mathcal{B}}

\newcommand{\cH}{\mathcal{H}}
\newcommand{\cS}{\mathcal{S}}

\newcommand{\cX}{\mathcal{X}}

\usepackage{amsthm}
\makeatletter
\newtheorem*{rep@theorem}{\rep@title}
\newcommand{\newreptheorem}[2]{%
	\newenvironment{rep#1}[1]{%
		\def\rep@title{#2 \ref{##1}}%
		\begin{rep@theorem}}%
		{\end{rep@theorem}}}
\makeatother

\newreptheorem{theorem}{Theorem}

\title{One-shot purity distillation with local noisy operations and one-way classical communication}
\author[1]{Sayantan Chakraborty\thanks{kingsbandz@gmail.com/sayantc@nus.edu.sg}}
\author[2]{Aditya Nema\thanks{aditya.nema30@gmail.com}}
\author[2]{Francesco Buscemi\thanks{buscemi@i.nagoya-u.ac.jp}}
\affil[1]{Centre for Quantum Technologies, National University of Singapore}
\affil[2]{Department of Mathematical Informatics, Nagoya University}

\date{\today}

\begin{document}

\maketitle

\begin{abstract}

Local pure states represent a fundamental resource in quantum information theory. In this work we obtain one-shot achievable bounds on the rates for local purity distillation, in the single-party and in the two-party cases. In both situations, local noisy operations are freely available, while in the two-party case also one-way classical communication can be used. In addition, in both situations local pure ancillas can be borrowed, as long as they are discounted from the final net rate of distillation. The one-shot rates that we obtain, written in terms of mutual information-like quantities, are shown to recover in the limit the asymptotic i.i.d.  rates of Devetak [PRA, 2005], up to first order analysis.

\end{abstract}

\section{Introduction}
Pure states play a crucial role in quantum theory and its new technologies, from quantum computing to quantum communication and quantum cryptography. This is true not only for practical purposes, where quantum effects typically reveal themselves only at very low (effective) temperatures, but also in the mathematics of quantum information theory, where techniques involving the purification of states and operations provide the starting point of virtually all theorems and calculations that we know of. In this latter context, pure states are typically introduced as ancillary registers, in tensor product with the system under consideration, ready to be used as workspace for some global unitary quantum operation, but sometimes they can be introduced surreptitiously, for example, when applying isometries embedding a quantum system into a larger space.

For the most part, in quantum information and quantum computation pure ancillas are treated as a resource that is freely available. However, moving from ``digital'' information to ``analog'' information, such as the information about the direction of a gyroscope, pure states become \textit{the} resource to account for, and cannot be treated as free anymore. This idea is particularly natural in thermodynamics. For example, Landauer \cite{Landauer} showed that one needs to do work in order to initialize an unknown state into a known pure state, a task that is known as \textit{erasure}. Conversely, Bennett et al. \cite{Bennet_etal_thermo}, resonating with the seminal work of Szilard \cite{Szilard}, showed that it is possible to extract work from a thermal bath if the system in initialized in a pure state.

Refs.~\cite{Oppenheim_puritydil,Horodecki_puritydilution,Synak-purity} are among the first to develop this idea in quantum Shannon theory, by introducing the framework of \textit{noisy operations}. Here the task is to start from a given mixed state and produce as many pure qubits as possible, by using only partial traces, unitary operations, and complete dephasings, i.e., projective measurements on an orthonormal basis. Devetak~\cite{Devetak_purity} and Krovi--Devetak~\cite{KroviDevetak} extended the task to a two-party scenario, with unbounded and bounded one-way classical communication, respectively. In comparison with previous analyses, Refs.~\cite{Devetak_purity,KroviDevetak} also allow the borrowing of pure ancillary states, as long as they are discounted from the final net rate. Under these constraints, Ref.~\cite{Devetak_purity} provides direct and converse coding theorems showing that the sum of the asymptotic rates at which both Alice and Bob can distill local pure states is given by
\[
\log d_A - H(A)_\rho + \log d_B - H(B)_\rho + \lim_{n \to \infty} \frac{1}{n} \max_{\Lambda^n} I(X^n;B^n)\;,
\]
where $\Lambda^n$ is a rank one POVM that maps the quantum system $A^{\otimes n}$ to the classical random variable $X^n$. 

The last term on the right-hand side of the above equation represents the excess rate with respect to what Alice and Bob would have obtained by acting in a purely local fashion (i.e., without communication), and it coincides with the one-way distillable common randomness of $\rho^{AB}$, namely, the maximum rate at which $\rho^{AB}$ can be used to extract bits of common randomness using one-way classical communications and local operations~\cite{Devetak_Winter}. Subsequently, Ref.~\cite{KroviDevetak} characterized the rate for the above bipartite case when the allowed classical communication was limited.

In this paper we provide achievable rates which serve as one-shot analogues of Devetak's original results. More details are provided in what follows.

\subsection{The Task}
We investigate the problem of producing pure states by performing some quantum operations on a given mixed state in the one-shot setting. We consider the following two scenarios of this problem: 

\begin{enumerate}
    \item In the first scenario, one party, say Alice, is provided with a single copy of some quantum state $\rho^A$ on system $A$. The task for Alice is to extract pure qubit states using only unitary operations on $A$ and classical randomness. Notice that this allows Alice to perform also dephasing channels, as these are classical mixtures of unitary operations. We call this task \emph{purity concentration}.
    \item In the second scenario, two parties, Alice and Bob possess the $A$ and $B$ sub-systems, respectively, of a given bipartite quantum state $\rho^{AB}$. They are allowed to perform local random unitary operations and to communicate via a dephasing (i.e., classical) channel. The task for Alice and Bob is to design a protocol using these resources such that together they can extract pure qubit states from the shared state $\rho^{AB}$. We call this task \emph{local purity distillation} but we will often refer to it as the \emph{distributed case}, and a protocol for it as a \emph{distributed protocol}.
\end{enumerate}

\begin{remark}
	In what follows, all logarithms are in base 2, and we implicitly assume the floor $\floor* d$ or the ceiling $\ceil* d$ (i.e., greatest/smallest integer less/grater than or equal to $d$), depending on the context, in any computation that evaluates the number of bits or qubits, or the dimension of some subspace, or the support of a probability distribution.
\end{remark}

\begin{definition}[purity concentration $\eps$-code]\label{def:e-purity-concentr}
Given a quantum state $\rho^A$ and a value $\eps\in[0,1]$, a purity concentration $\eps$-code consists of a $d_C$-dimensional ancilla system $C$ and a unitary $U^{AC\to A_pA_g}$ (the subscript $p$ stands for \emph{pure}, while $g$ stands for \emph{garbage}) such that
\[
\norm{\Tr_{A_g}[U (\rho^A\otimes\ketbra{0}^C) U^\dag]-\ketbra{0}^{A_p}}_1 \leq \eps\;.
\]
The rate of the code is defined as
\[
R\coloneqq \log d_{A_p}-\log d_C\;.
\]
A rate $R$ is defined to be $\eps$-achievable for purity concentration with respect to $\rho^A$ if there exists a purity concentration $\eps$-code with rate $R$.
\end{definition}

Notice how the size of the pure ancillary state borrowed at the beginning of the protocol is eventually discounted from the rate.

\begin{definition}[$\eps$-purity]
Given a state $\rho^A$, its $\eps$-purity, denoted by $\kappa_{\eps}(\rho^A)$ is defined as the supremum over all $\eps$-achievable rates for purity concentration.
\end{definition}

Before we formally define the distributed protocol, which is the main focus of this work, we wish to make some remarks about the resources which the protocol is allowed to use. Recall that in the LOCC paradigm the parties are allowed to use local operations and classical communication, with access to unlimited local ancilla. In contrast, the  CLOCC (closed local operations and classical communication) paradigm, introduced by Horodecki et al.~\cite{CLOCC, Devetak_purity} and described as a modification of the LOCC framework, does not allow the use of local ancilla. In this work, we extend CLOCC by adding to it two extra resources:
\begin{enumerate}
\item we allow the use of local pure ancillas, as long as their amount is discounted from the final rate of distillation;
\item we allow the use of local classical randomness, which also implies that noisy operations~\cite{Horodecki_puritydilution,GOUR20151} are considered free.
\end{enumerate}

\begin{remark}
	In order to avoid unnecessarily long expressions and duplicate notations, we adopt the convention that, given two operators $A$ and $B$, the Hermitian convolution $ABA^\dagger$ will be simply denoted as $A\cdot B$. The same notation, when necessary and no confusion arises, is extended to channels acting on operators.
\end{remark}

\begin{definition}[One-way local purity distillation $\eps$-code]
Given a bipartite quantum state $\rho^{AB}$, where the systems $A$ and $B$ belong to the two separated parties Alice and Bob, respectively, a local purity distillation $\eps$-code consists of:
\begin{enumerate}
	\item an ancillary system of dimension $d_C$;
    \item a unitary operator $U:AC\to A_pX$ on Alice's side;
    \item a dephasing channel $\mathcal{P}^{X\to X}$ with Kraus elements $\brak{\ketbra{x}}$ where $\brak{\ket{x}}$ is an orthonormal basis for the space $\mathcal{H}_X$;
    \item a unitary operator $V:XB\to B_pB_g$ on Bob's side.
\end{enumerate}
The above code should satisfy the condition that
\[
\norm{\Tr_{B_g}\Big[V\cdot \mathcal{P}\cdot U\cdot (\rho^{AB}\otimes \ketbra{0}^C)\Big]-\ketbra{0}^{A_p}\otimes \ketbra{0}^{B_p}}_1\leq \eps
\]
The rate $R$ of the code is defined to be
\[
R^{\to}\coloneqq \log d_{A_p}+\log d_{B_p}-\log d_C
\]
A rate $R^\to$ is said to be $\eps$-achievable for one-way local purity distillation if there exists a one-way local purity distillation $\eps$-code with rate $R^\to$.
\end{definition}

\begin{definition}[One-way $\eps$-distillable local purity]
The supremum over all $\eps$-achievable rates $R^{\to}$ for one-way local purity distillation, denoted by $\kappa^{\to}_{\eps}(\rho^{AB})$, is defined as the one-way $\eps$-distillable local purity of state $\rho^{AB}$.
\end{definition}

\subsection{Results}\label{subsec:Results}

We derive achievable rates for both the aforementioned tasks of purity concentration and one-way purity distillation with $\eps$ error in the one shot setting. Our rate expressions are derived in terms of one-shot entropic quantities which nonetheless converge to the known Shannon-theoretic bounds. We prove the following theorems (for the formal definition of the entropic quantities appearing below, see Section~\ref{sec:entropies} of the Appendix).

\begin{theorem}[Purity concentration]\label{thm:purityconc}
Given a state $\rho^A$ and a value $\eps\in[0,1]$, we have
\[
\kappa_{3\sqrt{\eps}}(\rho^A)\ge \log d_A - \widetilde{H}_{\max}^{\eps}(A)_{\rho^A}\;.\label{eq:local}
\]
\end{theorem}

\begin{theorem}[One-way local purity distillation] \label{thm:purity_distillation}
Given a bipartite state $\rho^{AB}$ and a value $\eps\in(0,1)$, 
\[
\kappa_{\eps'}^\to(\rho^{AB})\ge \log d_Ad_B-\widetilde{H}_{\max}^{O(\eps^2)}(A)_{\rho^A}-\widetilde{H}_{\max}^{\eps}(B)_{\rho^B}+D^{\to}_{\eps}(\rho^{AB})+O(\log\eps)-O(1)\;,
\]
where
\begin{itemize}
	\item $\eps'$ is a suitable rational power function of $\eps$;
	\item $D^\to_{\eps}(\rho^{AB}):= \max_{\Lambda} I_{H}^{\eps_0}(X:B)_{\I^B\otimes \Lambda^A(\rho^{AB})}$,
	where $\eps_0=O(\eps^{1/4})$, and the maximization is over all rank-one POVMs $\Lambda^{A\to X}$ on the system $A$.
\end{itemize}
The above bound can be achieved with an amount of classical communication of at most $H_{\max}^{O(\eps^2)}(A)_{\rho^A}+O(\log \frac{1}{\eps})$.
\end{theorem}

\begin{remark}
    As we will show later, in the case when a large number of  independent copies of the resources are available and in the limit of the error $\eps\to 0$, the bound in \cref{thm:purity_distillation} recovers the optimal bound for this problem, as originally shown by Devetak \cite{Devetak_purity}.
\end{remark}

\begin{remark}
	The maximization in the theorem above can be restricted to rank-one POVMs without loss of generality. This is because any POVM of rank larger than one can always be obtained as the post-processing of another rank-one POVM~\cite{clean-POVMs}. Hence, by the data-processing inequality, such post-processed POVM would lead to a smaller one-shot mutual information.
\end{remark}

The rest of the paper is organized as follows.
In \cref{sec:purity_concentration} we state and prove the one-shot version of the local single system protocol for purity distillation.
In \cref{sec:overview_techniques} we present a high level description of the main bipartite purity distillation protocol, highlighting in particular the issues that arise with a straightforward adaptation of Devetak's asymptotic protocol in the one-shot setting.
In \cref{sec:Technicaltools} we state the technical lemmas used in the purity distillation protocol, though we defer their proofs to \cref{sec:Main_proofs}.
In \cref{sec:MainProtocol} we state and prove the one-shot purity distillation protocol for the bipartite case.
In \cref{sec:Conclusion} we conclude the main text, by summarizing the main results and the open problems to be pursued in future. Two appendices follow: in \cref{sec:Preliminaries} we mention the notation and preliminary mathematical facts used throughout the paper. 
In \cref{sec:entropies} we mention the entropic quantities that we use, together with their properties and asymptotic i.i.d. behavior.

\section{Overview of the purity concentration protocol} \label{sec:purity_concentration}

In this section, as a preparatory example, we describe the one-shot purity concentration theorem. (The distributed one-shot purity distillation theorem is described later, in \cref{sec:overview_techniques}.)

Even though we will provide the formal definition later (see \cref{def:smooth_supp_maxentropy}), we need here a new one-shot smoothed entropic quantity, which we refer to as the \emph{smoothed-support maximum entropy} and denote it as $\widetilde{H}_{\max}^\epsilon(A)_{\rho^A}$. For any state $\rho^A$, the corresponding $\widetilde{H}_{\max}^\epsilon(A)_{\rho^A}$ is defined as the logarithm of the dimension of $\operatorname{supp} (\rho'^A)$, where $\rho'^A$ is the state obtained by zeroing out the smallest eigenvalues of $\rho^A$ which sum up to at most $\eps$, for any $\eps\ge0$. This quantity plays a central role in determining the rates of purity concentration and distillation.

Recall from \cref{def:e-purity-concentr} that purity concentration involves some unitary operation on the input state $\rho^A$ and then discarding a part of the output system, such that the register which is left, i.e. $A_p$, contains a state which is close to a pure state $\ket{0}^{A_p}$. To do this, the main idea is to discard the smallest eigenvalues of $\rho^A$, which add up to at most $\eps$. The eigenvectors which are left then span a space of dimension $2^{\widetilde{H}_{\max}^{\eps}(A)}$, but are embedded in the larger system $A$. Let us refer to these eigenvectors as `good'. This embedding necessarily requires that the eigenvectors be padded with $0$'s on the extra coordinates which are not required to specify them. Thus, we can relabel each of these good eigenvectors with a vector of dimension $2^{\widetilde{H}_{\max}^{\eps}(A)}$ tensored with the unit vector $\ket{0}$. These unit vectors then necessarily belongs to a space of dimension $\abs{A}/2^{\widetilde{H}_{\max}^{\eps}(A)}$. We make these ideas rigorous below.

We will first require the following fact, which can be found in \cite[Lemma~1]{Devetak_purity}, but we prove it here once more for completeness:

\begin{fact}\label{fact:local}
Consider a vector space $A\cong A_g\otimes A_p$, with $\dim A_g=d_1$ and $\dim A_p=d_2$, a state $\rho$ on $A$, and a projector $\Pi$ with rank equal to $d_1$. If $\Tr[\Pi\rho]\geq 1-\eps$, then there exists a unitary $U$ on $A$, a (normalized) state $\tilde\rho$ on $A_g$, and a pure state $\ket{0}\in A_p$ such that
\[
\norm{U\rho U^{\dagger} - \tilde\rho\otimes \ketbra{0}}_1 \leq 3\sqrt{\eps}\;.
\]
\end{fact}

\begin{proof}
From the condition $\Tr[\Pi\rho]\geq 1-\eps$, by virtue of the ``gentle measurement'' lemma \ref{fact:Gentle_operator_lemma}, we know that $\norm{\rho-\Pi\rho\Pi}_1\le 2\sqrt{\eps}$. In particular, this implies that $\norm{\rho-\frac{\Pi\rho\Pi}{\Tr[\Pi\rho]}}_1\le\norm{\rho-\Pi\rho\Pi}_1+\norm{\Pi\rho\Pi-\frac{\Pi\rho\Pi}{\Tr[\Pi\rho]}}_1 \le 2\sqrt{\eps}+\eps\le 3\sqrt{\eps}$.
Let now $\{ \ket{v_i}\}$ be an orthonormal basis for the support of $\Pi$. Since $\dim\mathrm{supp}\Pi=d_1$, there exists a unitary $U$ on $A$ such that
\[
U\ket{v_i}=\ket{\tilde v_i}\otimes\ket{0}\;,
\]
for a choice of orthonormal vectors $\ket{\tilde v_i}\in A_g$ and a pure state $\ket{0}\in A_p$. Then, since $[\Pi,\Pi\rho\Pi]=0$, we have that
\[
U\frac{\Pi\rho\Pi}{\Tr[\Pi\rho]} U^\dagger=\tilde\rho\otimes\ketbra{0}\;,
\]
for some normalized state $\tilde\rho$. Finally, as a consequence of the invariance of the trace-norm under unitary transformations,
\begin{align*}
    \norm{U\rho U^\dag - \tilde\rho\otimes\ketbra{0}}_1=\norm{\rho-\frac{\Pi\rho\Pi}{\Tr[\Pi\rho]}}_1\le 3\sqrt{\eps}\;,
\end{align*}
as claimed. \\
\end{proof}

Using \cref{fact:local}, it is easy now to prove \cref{thm:purityconc} mentioned in \cref{subsec:Results} as  follows:

\begin{proof}[\bf Proof of \cref{thm:purityconc}:]
	Given $\rho^A$ and $\eps\in[0,1]$, let us introduce an ancillary system $C$ such that $d_C=2^{\widetilde{H}_{\max}^{\eps}(A)_\rho}$. This step is necessary in order to make dimensions (which are integer numbers) factorize nicely, so that the protocol is deterministic (i.e., unitary). 
	
	Denote by $\Pi^A$ the projector onto the support of $\rho'^A$, namely, the sub-normalized state obtained by zeroing out the smallest eigenvalues of $\rho^A$ which add to less than or equal to $\eps$. Notice that $\Tr[\Pi^A\rho^A]\geq 1-\eps$ and $\Tr[\Pi^A]=d_C$. Let us now define the extended state
	\[
	\rho^{AC}:=\rho^A\otimes\ketbra{0}^C\;.
	\]
	It is clear that $\widetilde{H}_{\max}^{\eps}(AC)_\rho =\widetilde{H}_{\max}^{\eps}(A)_\rho$. Analogously, let us define the extended projector $\Pi^{AC}=\Pi^A\otimes\ketbra{0}^C$. Clearly, $\Tr[\Pi^{AC}\ \rho^{AC}]=\Tr[\Pi^A\ \rho^A]\ge 1-\eps$.
	
	Thus, by invoking \cref{fact:local}, we see that there exists a unitary operator from $A\otimes C$ to $A_g \otimes A_p\cong A\otimes C$, which satisfies
	\begin{align*}
	\norm{U(\rho^A\otimes\ketbra{0}^C) U^\dag - \tilde\rho^{A_g}\otimes\ketbra{0}^{A_p}}_1\le 3\sqrt{\eps}\;,
	\end{align*}
	implying that
	\begin{align*}
	\kappa_{3\sqrt{\eps}}(\rho^A\otimes\ketbra{0}^C)&\ge\log\dim A_p\\&=\log (d_Ad_C)-\widetilde{H}_{\max}^{\eps}(A)\;,
	\end{align*}
	which reduces to the statement of the  theorem once we discount the amount $\log d_C$ of purity that we borrowed at the beginning.
	
	\begin{remark}
		In order to provide the reader with an intuitive understanding of why we need to borrow pure ancilla qubits, let us write the trimmed state $\rho'^A$ in block form
		\[
		\rho'^A=\begin{pmatrix}
			C & 0\\
			0 & 0
		\end{pmatrix}\;,
		\] 
		where the dimension of the block $C$ is $d_C$. We then summon a pure ancilla state $|0\rangle^C$ of the same dimension, and act on the joint state with a unitary operator that \textit{swaps} the $R$ block with the ancilla. This can be done, since the two have the same dimension. In this way, at the end we are left with a $d_A$-dimensional pure state, giving $\log d_A$ bits of purity, and from this we discount the size of the ancilla that we borrowed, which is exactly $\log d_C$.
	\end{remark}
\end{proof}

\section{Overview of the purity distillation protocol} \label{sec:overview_techniques}

In this section we describe the main ideas behind our results, which we use to generalize Devetak's results to the one-shot setting. This section is meant to serve as a road map for the proofs, which appear later in \cref{sec:Main_proofs}. We make use of the purity concentration protocol described in \cref{sec:purity_concentration} as a subroutine. 

We being by recalling the setup: Alice and Bob share the $A$ and $B$ parts, respectively, of a bipartite state $\rho^{AB}$. They are allowed to use local unitaries and a dephasing (i.e., classical) channel to communicate. They can also borrow local pure ancilla, but these will be discounted from the final rate. An obvious  protocol, requiring no communication, is obtained if Alice and Bob simply enact the concentration protocol locally on $A$ and $B$ systems respectively. In this way, they can extract local pure states at the rate 
\begin{equation} \label{eq:Naiverate}
\log d_A-\widetilde{H}_{\max}^{\eps}(A)+\log d_B-\widetilde{H}_{\max}^{\eps}(B)\;.
\end{equation}
However, the above is not optimal, as shown in the following example.

\begin{example}
Consider the maximally correlated state
\[
\rho^{AB}\coloneqq\frac{1}{2}\ketbra{0}^A\otimes \ketbra{0}^B+\frac{1}{2}\ketbra{1}^A\otimes \ketbra{1}^B
\]
and set $\eps=\frac{1}{4}$. Clearly, we cannot discard any eigenvalues from the marginals $\rho^A$ and $\rho^B$, and hence the two concentration protocols on the $A$ and $B$ systems together produce no pure states. However, if Alice were to send the system $A$ to Bob via a dephasing channel with operational elements $\ketbra{0}$ and $\ketbra{1}$, then Bob could apply the following unitary:
\[
\ketbra{0}^A\otimes \I^{B}+\ketbra{1}^A\otimes X^B
\]
where $X$ is the quantum NOT (i.e., Pauli $X$) operator. Clearly, this allows Bob to extract one qubit pure state. Thus, this example demonstrates that introducing classical communication between the two parties can lead to strictly better rates. 
\end{example}

Notice that the key idea used in the above example is to leverage the \emph{classical} correlations between the systems $A$ and $B$. However in general the $A$ and $B$ systems shared by Alice and Bob will share quantum correlations. Then the idea is that Alice measures her system using a POVM to create a classical-quantum (cq) state, and then send the contents of the classical register created by this measurement to Bob. The hope is that by doing some measurement on his system, Bob should be able to distinguish among the contents of the classical register. If he is able to do this, then he can appropriately map the contents of the classical register to a pure state $\ket{0}$. However, there are several subtle issues that needs to be addressed.

\subsection{A bad protocol}\label{Protocol:dud}
To make things precise, we first consider a protocol which does not work. However, studying this bad protocol eventually leads us to the correct answer. To that end, consider the following `dud' protocol:
\begin{enumerate}
	
    \item Alice has some rank-one POVM $\brak{\Lambda_x}$ with elements labelled by some finite set $\mathcal{X}=\{x\}$. Let $\ket{\varphi_{\rho}}^{ARB}$ be a purification of the state $\rho^{AB}$. (In fact, we would not need $R$ at this point, but working with pure states makes equations more compact.) Now consider the isometry from $A$ to $XA$:
    \[
    V:=\sum_{x}\ket{x}^X\sqrt{\Lambda_x^A}\;.
    \]
    Alice can simulate the action of this isometry on the system $A$ by borrowing the pure state $\ket{0}^{X}$ and then acting an appropriate unitary on the $AX$ register of the state 
    \[
    \ket{0}^{X}\ket{\varphi_{\rho}}^{ARB}
    \]
    Since the POVM is rank-one, this action produces the state
    \[
     \sum_x \sqrt{P_X(x)}\ket{x}^X\ket{\psi_x}^A\ket{\phi_x}^{RB}\;,
    \]
    with $P_X(x)=\Tr[\Lambda_x^A\ \rho^A]$.
    
    \item Next, Alice can simply condition on each $x$ in the system $X$ and map each $\ket{\psi_x}^A$ to the state ${\ket{0}}^A$ by applying the following controlled unitary
    \[
    \sum_x\ketbra{x}^X\otimes U_x^{A\to A}\;,
    \]
    where for all $x\in\mathcal{X}$
    \[
    U_x\ket{\psi_x}^A= \ket{0}^A\;.
    \]
    This allows Alice to recover pure states at the rate $\log d_A$. However, recall that she borrowed $\log d_X$ amount of pure ancilla. We will account for this later by subtracting it from the overall rate.
    
    \item At this point of the protocol, the joint state is
    \[
    \ket{0}^A\otimes\sum_x \sqrt{P_X(x)}\ket{x}^X\ket{\phi_x}^{RB}\;,
    \]
    with system $X$ still with Alice.
    
    \item Next, Alice applies a local dephasing channel on the system $X$, so that the state becomes
    \[
    \ketbra{0}^A\otimes\sum_x P_X(x)\ketbra{x}^X\otimes\ketbra{\phi_x}^{RB}\;.
    \]
    
    \item Next, Alice can apply a second time the local purity concentration protocol on subsystem $X$, obtaining two subsystems, $X_g$ and $X_p$, and distilling further $\log\dim X_p$ bits of local purity. Note that, by construction,
    \[
    \log d_{X_p}= \log d_X-\widetilde{H}_{\max}^{\eps}(X)\;.
    \]
    She does this by considering the subset $\mathcal{S}\subseteq\mathcal{X}$ obtained by throwing away those $x$'s which correspond to the smallest probabilities $P_X(x)$ which add up to at most $\eps$. Thus, the net rate of pure states distilled so far is
    \[
    \log d_A+\log d_{X_p}-\log d_X = \log d_A-\widetilde{H}_{\max}^{\eps}(X)\;,
    \]
    where we discounted the amount $\log d_X$ that Alice borrowed in advance.
    
    \item At this point in the protocol, we note that the global state on $X_gRB$ is close to the state
    \[
  \frac{1}{\sum_{x'\in \mathcal{S}}P_X(x')}\sum_{x'\in \mathcal{S}}P_X(x')\ketbra{x'}^{X_g}\otimes\ketbra{\phi_{x'}}^{RB}\;.
    \]

\item Alice will now send to Bob the contents of the $X_g$ register through the dephasing channel. Note that we must fix the basis in which the dephasing channel acts to the computational basis of the system $X_g$. Again, for simplicity of the notation, let us rename all those $x'$'s which correspond to $x$'s in $\mathcal{S}$ by classical symbols $y$, and rename the distribution $P_X$ conditioned on $\mathcal{S}$ as $P_Y$. Thus, the state on $YB$ (ignoring $R$) may now be written as
\[
\sum_{y}P_Y(y)\ketbra{y}^Y\otimes \phi_{y}^{B}\;,
\]
where $\phi^B_y$ denotes the mixed state $\Tr_R[\ketbra{\phi}_y^{RB}]$.

\item At this stage, Bob would like to measure the $B$ system so as to produce a guess about the contents of the system $Y$, and apply a controlled unitary on $Y$ based on this guess to map it to a pure state. This strategy will not work however, since, in general, Bob is not able to distinguish among all $\abs{\mathcal{S}}$ states $\phi_y^B$. Instead, Bob hashes the $y$'s into bins randomly. The hope is that if the size of each bin is small enough, then conditioned on the bin index, Bob will be able to distinguish among the quantum states associated with that bin. For this to work, we must ensure that each bin contains the same number of symbols and that for each bin, there exists a POVM which can distinguish amongst the $\phi_y$ states corresponding to the symbols in that bin.

\item The above claim is technically involved, since usual binning strategies (using random binning or $2$-universal hash functions) do not ensure that each bin has the same number of elements. This is especially hard in the one-shot setting, since we cannot leverage concentration bounds. To overcome this issue, in \cref{claim:permutation_hash} we prove that such an appropriate binning strategy indeed exists, and that, corresponding to each bin, there exists an appropriate decoding POVM. Our new lemma uses \emph{random permutations} instead of $2$-universal hash functions for the binning, which turns out to be suitable for our purposes. Details can be found in \cref{sec:Technicaltools}. We show that each bin can be at most of size 
\[
2^{I_H^{\eps}(X:B)_{\rho^{XB}}}\;.
\]
(The formal definition of the above entropic quantity will be given in \cref{sec:entropies}; for the moment, suffice it to say that $I_H^{\eps}(X:B)_{\rho^{XB}}$ is a quantum mutual information--like quantity.) Note that this implies that Bob can only turn the \emph{intra-bin} indices into pure states, and not the bin indices themselves, since he can only condition on each bin index and then distinguish among the contents of that bin. The rate at which Bob produces pure states at this stage is then
\[
I_{H}^{\eps}(X:B)_{\rho^{XB}}\;.
\]

\item One can also show that the decoding strategy used above does not perturb the state on the system $B$ too much. Thus Bob can finally use the purity concentration protocol of what is left locally, that is, \cref{thm:purityconc} on his system $B$, and finish the protocol, thus producing pure states at a net rate of
\[
I_{H}^{\eps}(X:B)_{\rho^{XB}}+\log d_B-\widetilde{H}_{\max}^{\eps}(B)\;,
\]
and a total rate of
\[
\log d_A-\widetilde{H}_{\max}^{\eps}(X)+
I_{H}^{\eps}(X:B)_{\rho^{XB}}+\log d_B-\widetilde{H}_{\max}^{\eps}(B)\;.
\]
\end{enumerate}

\subsection{Issues with the bad protocol and a way around them}

An obvious issue with the above protocol is that we have no way to bound the term $\widetilde{H}_{\max}^{\eps}(X)$ in terms of entropic quantities computed with respect to the original state $\rho^{AB}$. To remedy this situation, we will use the measurement compression theorem.

The original measurement compression theorem, due to Winter \cite{Winter_meascomp}, takes as input $n$ tensor copies of the state $\rho^{AB}$ and the POVM $\Lambda^A=\brak{\Lambda^A_x}$ and produces as output an integer $K$ and a class of POVMs $\brak{\Gamma(k)=\brak{\Gamma_{\ell}(k)}~|~k\in [K]}$, where each POVM $\Gamma(k)$ is defined on the system $A^{\otimes n}$ and induces a distribution $P_{L|k}$ upon measuring $\rho^{\otimes n}$. When quantum side-information $B$ is available at the receiver, an extension of Winter's protocol~\cite{Wilde_etal_measurementcomp} states that, as long as 
\begin{align*}
    \frac{1}{n}\log L &> I(X:RB)\;, \\
    \frac{1}{n} \left(\log K + \log L \right)&> H(X)\;,
\end{align*}
the class of POVMs $\brak{\Gamma(k)}$ \emph{faithfully simulates} the action of $\Lambda^{\otimes n}$ on $\rho^{\otimes n}$. What this means is, the distribution produced when $\Lambda^{\otimes n}$ is used to measure $\rho^{\otimes n}$ is almost statistically indistinguishable from the distribution produced by the following experiment:

\begin{enumerate}
    \item pick $k\in [K]$ uniformly at random;
    \item measure $\rho^{\otimes n}$ with POVM $\Gamma(k)$ and obtain outcome $\ell$;
    \item post-process $k$ and $\ell$ into an $x$.
\end{enumerate}

\noindent The entropic quantities above are computed with respect to the following state:
\[
\sum\ketbra{x}^X\otimes \Tr_A\left[\Lambda_x^A\left(\ketbra{\varphi_{\rho}}^{ABR}\right)\right]\;,
\]
where $\ket{\varphi_{\rho}^{ABR}}$ is a purification of $\rho^{AB}$.
Devetak's idea \cite{Devetak_purity}, roughly speaking, is to use one of the smaller POVMs produced by the measurement compression theorem to measure $\rho^{\otimes n}$ instead of using $\Lambda^{\otimes n}$. This works since the set of `good' outcomes for any $\Gamma(k)$ is of size at most 
\[
2^{nI(X:RB)}\leq 2^{nH(A)}\;.
\]
This allows Devetak to bound the size of the classical register that Alice needs to send to Bob. However, there are further technical issues here.

A closer look at the measurement compression theorem shows us that each \emph{compressed} POVM $\Gamma(k)$ consists of a set of `good' outcomes, which is of size $2^{nI(X:RB)}$ and a `bad' outcome $\bot$, which occurs with probability at most $\eps$. Let us denote the POVM element which corresponds to the outcome $\bot$ as 
\[
\Gamma_{\bot}(k)\coloneqq \I^A-\sum_{\ell}\Gamma_{\ell}(k)
\]
Observe that $\Gamma_{\bot}(k)$ will in general not be a rank-one matrix, which in turn means that $\Gamma(k)$ will not be a rank-one POVM. Recall that we require our POVMs to be rank-one. Instead, one has to consider the rank-one elements in the eigendecomposition of $\Gamma(k)$ and add them as individual POVM elements. This defeats the purpose of using the measurement compression theorem to bound the number of outcomes in the first place!

To address this issue, Devetak heavily relies on tools which are only available when $n$ tensor copies of the system $A$ are given, i.e., the properties of typical subspaces that arise from the asymptotic i.i.d. assumptions. He shows that for sufficiently large $n$, the space $A^{\otimes n}$ can be decomposed in $A_1\otimes A_2$, in such a way that the entropy of Alice's state restricted to the subspace $A_1$ is at most $n \eps$, and the compressed POVM $\Gamma(k)$ is rank-one on $A_2$, while incurring in a vanishingly small error. Since one needs at most $d$-many elements to complete the description of the POVM, one can then bound the number of outcomes by $2^{nH(A)}$. This idea is not straightforward to implement in the one-shot setting.

A further issue is that even though the compressed POVM $\Gamma(k)$ has fewer good outcomes, it needs to preserve the classical correlations between the systems $A^{\otimes n}$ and $B^{\otimes n}$. What this means is that, supposing $\Gamma(k)$ produces the classical system $Y$ as output, the mutual information between $Y$ and $B^{\otimes n}$ should be at least the mutual information $I(X:B)$. Devetak used a simple derandomization argument to show that such a $\Gamma(k)$ indeed exists.

As mentioned earlier, the above arguments are not easy to emulate in the one-shot setting. Hence our approach is slightly different from Devetak's, in the sense that we do not rely on concentration arguments to show that there exists a good compressed POVM $\Gamma(k)$ with a small number of good outcomes, which also preserves the classical correlations. Firstly, we require a one-shot measurement compression theorem. This was recently proved by Chakraborty, Padakandla and Sen in \cite{ChakrabortyPadakandlaSen_22}. Next, we show in \cref{lem:povm} below that there exists at least one sub-normalized POVM, which preserves the classical correlations between the two systems, as measured in terms of the smoothed hypothesis testing mutual information. This step is hard since chain rules, readily available in the case of the Shannon mutual information, are not known for this quantity.

Next, we extend this sub-POVM to a full rank-one POVM $\Gamma(k)$ by extending the set of outcomes using the eigendecomposition of the POVM element $\Gamma_{\bot}(k)$. Note that this blows up the set of outcomes to a set which contains at least as many indices as the dimension of the underlying space. This is because of the additional outcomes which together correspond to the bad outcome $\bot$. However, we mitigate this issue by leveraging the fact that all these bad outcomes together have probability at most $\eps$. The key idea is that instead of using the set of indices with the lowest probabilities which add up to $\eps$ for the extracting the pure states locally at Alice's end (as in step~5 of the dud protocol in \cref{Protocol:dud}), we instead use the set of \emph{bad} outcomes of our POVM. This allows us to distill local purity at Alice's end at the rate 
\[
\log d_A -I_{\max}^{\eps}(X:RB)\;,
\]
where the quantity $I_{\max}^{\eps}(X:RB)$ can be bounded from above by $H_{\max}^{\eps}(A)$, and hence also by $\widetilde{H}_{\max}^{\eps}(A)$, as required. Details can be found in \cref{sec:MainProtocol}.

\begin{remark}
    It may seem to the reader that one need not have considered the original POVM $\Lambda$ at all, since in the final protocol we use only the compressed POVM provided by the one-shot measurement compression theorem. However, one should note that $\Lambda$ is the maximizer for the smooth max mutual information $I_{\max}^{\eps}(X:RB)$, which quantifies the amount of classical correlation between the two systems. However, one needs to introduce the compressed POVM since using $\Lambda$ off the shelf does not allow us to bound the classical communication from Alice to Bob with a meaningful quantity.
\end{remark}

\section{Technical lemmas}\label{sec:Technicaltools}

In this section we describe the main technical lemmas used in achieving one-shot rate for purity distillation protocol. Since the proofs are technically involved, we defer them to \cref{sec:Main_proofs}.

\begin{remark}
    Throughout this paper, we will use the notation $\eps_0$ to denote some constant time the fourth root of $\eps$, which we use to denote the error bounds in most of our theorems, i.e.,
    \[
    \eps_0\coloneqq O(\eps^{{1/4}})
    \]
    This means that we will often abuse notation, and denote two different quantities such as $2\eps^{1/4}$ and $100\eps^{1/4}$, by the same notation, namely $\eps_0$. This will allow us to present the proofs in a much cleaner manner. In a similar vein, we will sometimes use the notation $\eps'$ to denote $O(\eps)$. 
\end{remark}

\subsection{Choosing a POVM}

\begin{lemma}\label{lem:povm}
Given a bipartite state $\rho^{AB}$ and a rank-one POVM $\brak{\Lambda_x^A}$ with outcomes in the set $\mathcal{X}$,  consider the post measurement state
\[
\rho^{XRB}\coloneqq \sum_x \ketbra{x}^X\otimes \Tr_A\left[(\Lambda_x^A\otimes\I^{RB})\ \ketbra{\varphi_{\rho}}^{ARB}\right]
\]
where $\ket{\varphi}^{ARB}$  is a purification of $\rho^{AB}$. Then, there exists a rank-one POVM $\brak{\Tilde{\Lambda}_y^A}$ with outcomes in the set $\mathcal{Y}$ such that:
\begin{enumerate}
    \item for any $\eps>0$, there exists a subset $\mathcal{S}\subset \mathcal{Y}$ such that
    \begin{align*}
		\abs{\mathcal{S}}&\leq 2^{I_{\max}^{\eps}(X:RB)_{\rho^{XRB}}}\\
        \intertext{and}
        \Pr_{P_Y}[\mathcal{S}] &\geq 1-\eps_0\;, 
    \end{align*}
    where $P_Y$ is the distribution induced by $\Tilde{\Lambda}$ on $\mathcal{Y}$ upon measuring $\rho^{A}$;
    \item denoting by $\Pi_{\mathcal{S}}^{Y}$ the projector onto the space spanned by the vectors corresponding to the elements in $\mathcal{S}$ and defining the corresponding projected and renormalized state as
    \[
    \sigma^{YRB}\coloneqq \frac{1}{\Tr[\Pi_{\mathcal{S}}\ \Tilde{\Lambda}(\varphi_{\rho})]} \Pi_{\mathcal{S}}^{Y}\cdot \Tilde{\Lambda}^{A}(\varphi_{\rho}^{ARB})\;,
    \]
    then
    \[
    I_H^{\sqrt{\eps_0}}(Y:B)_{\sigma}\geq I_{H}^{\eps_0/2}(X:B)_{\rho}-O(1)+O(\log(1-\eps_0))\;.
    \]
\end{enumerate}
\end{lemma}

\begin{remark}
Note that we place no restrictions the size of the set $\mathcal{Y}$.
\end{remark}

\begin{proof}
	See \cref{sub:proof_lem:POVM}.
\end{proof}


\subsection{Dividing the domain}

\begin{lemma} \label{lem:slepianwolf}
Given the control state
\[
\rho^{XB}= \sum_xP_X(x)\ketbra{x}^X\otimes \rho_x^B
\]
and a value $\eps\in(0,1)$, there exists a bijection $\sigma: \mathcal{X}\to [M]\times [N]$, such that:
\begin{enumerate}
    \item $M \times N = \abs{\mathcal{X}}$\;; 
    \item $\log N < I_{H}^{\eps}(X:B)+2\log \eps$\;;
    \item suppose the state after applying the bijection is given by
    \[
    \sigma^{MNB}:=\sum_{m,n} P_{MN}(m,n) \ketbra{m,n}^{MN}\otimes \rho_{mn}^B\;;
    \]
    then there exists, for all $m\in [M]$, a POVM $\brak{\Theta_{n}(m)}$ with outcomes labeled by $n\in[N]$,  such that
    \[
    \sum_{m,n}P_{MN}(m,n)\norm{\rho^{B}_{mn}-\sqrt{\Theta_n(m)}\rho^{B}_{mn}\sqrt{\Theta_n(m)}}_1 \leq \eps_0\;,
    \]
\end{enumerate}
\end{lemma}

\begin{proof}
See \cref{subsec:proof_lem:slepianwolf}.	
\end{proof}

\begin{corollary}\label{corol:bobsprotocol}
Given the state 
\[
 \sigma^{MNB}= \sum_{m,n} P_{MN}(m,n) \ketbra{m,n}^{MN}\otimes \rho_{mn}^B
\]
as in \cref{lem:slepianwolf}, there exists a unitary $W^{MNB}$ such that both conditions, i.e.,
\[
\norm{\Tr_{MB}\left(W^{MNB}\cdot \sigma^{MNB}\right)-\ketbra{0}^N}_1 \leq \sqrt{\eps_0}\;,
\]
and
\[
\norm{\Tr_{MN}\left(W^{MNB}\cdot \sigma^{MNB}\right)-\sum_{m,n}P_{MN}(m,n)\rho^B_{mn}}_1 \leq \eps_0\;,
\]
both hold simultaneously.
\end{corollary}

\begin{proof}
See \cref{subsec:proof_corol:bobsprotocol}.	
\end{proof}

\section{The two-party purity distillation protocol}\label{sec:MainProtocol}

In this section, we describe the main purity distillation protocol and prove the achievable one-shot rate as stated in \cref{thm:purity_distillation}. For the sake of convenience, we recall the statement:

\begin{reptheorem}{thm:purity_distillation}[One-way local purity distillation] 
Given a bipartite state $\rho^{AB}$ and a value $\eps\in(0,1)$, 
\[
\kappa_{\eps'}^\to(\rho^{AB})\ge \log d_Ad_B-\widetilde{H}_{\max}^{O(\eps^2)}(A)_{\rho^A}-\widetilde{H}_{\max}^{\eps}(B)_{\rho^B}+D^{\to}_{\eps}(\rho^{AB})+O(\log\eps)-O(1)\;,
\]
where
\begin{itemize}
	\item $\eps'$ is a suitable rational power function of $\eps$;
	\item $D^\to_{\eps}(\rho^{AB}):= \max_{\Lambda} I_{H}^{\eps_0}(X:B)_{\I^B\otimes \Lambda^A(\rho^{AB})}$, where $\eps_0=O(\eps^{1/4})$, and the maximization is over all rank-one POVMs $\Lambda^{A\to X}$ on the system $A$.
\end{itemize}
The above bound can be achieved with an amount of classical communication of at most $H_{\max}^{O(\eps^2)}(A)_{\rho^A}+O(\log \frac{1}{\eps})$.
\end{reptheorem}

 \begin{proof}
The proof of the above theorem consists in the following protocol:
\begin{enumerate}
    \item Alice and Bob start with the $A$ and $B$ parts of the state $\rho^{AB}$ in their possession respectively. In the first step, Alice applies the rank-one POVM $\brak{\Tilde{\Lambda}_y^A}$ given by \cref{lem:povm} on her system $A$ coherently. What this means is that Alice borrows $\log \abs{\mathcal{Y}}$ amount of ancilla and applies the isometry
    \[
      V_1^{A \to YA}:=\sum_y \ket{y}^Y\sqrt{\Tilde{\Lambda}_y}^{A}
    \]
    on the system $A$. Considering a purification $\ket{\varphi_{\rho}}^{ARB}$ of $\rho^{AB}$, this produces the global state
    \begin{align*}
        (V_1 \otimes \I^{RB}) \ket{\varphi_{\rho}}^{ARB}= \sum_y \sqrt{P_Y(y)}\ket{y}^Y\ket{\psi_y}^A\ket{\phi_y}^{RB}\;.
    \end{align*}
    Note that this is only possible since the elements $\Tilde{\Lambda}_y$ are rank-one. Notice that $\brak{\Tilde{\Lambda}_y^A}$ is the compressed POVM as opposed to the POVM mentioned in step 1 of the dud protocol of \cref{Protocol:dud}. Alice then applies a controlled unitary
    \[
    U_1^{AY \to AY}:=\sum_y\ketbra{y}^{Y}\otimes U^{A\to A}_y
    \]
    where, for each $y\in \mathcal{Y}$, $U^{A\to A}_y$ is a unitary such that
    \[
    U^{A\to A}_y\ket{\psi_y}^A = \ket{0}^A\;.
    \]
    This step yields $\log d_A$ amount of purity while using $\log \abs{\mathcal{Y}}$ amount of purity.
    \item In the next step, Alice dephases the $Y$ system, i.e., she measures it in the computational basis to create the global state
    \[
    \tau^{YB}\coloneqq \sum_y P_Y(y)\ketbra{y}^Y\otimes \rho_y^B
    \]
    where we have ignored the system $R$. Note that in this case we define $\rho_y^B$ for every $y\in \mathcal{Y}$ as the reduced state on $B$ conditioned on $y$. Let $\mathcal{S}\subset \mathcal{Y}$ be the set of high probability given by \cref{lem:povm} and let $\Pi^Y_{\mathcal{S}}$ be the projector onto the span of the computational basis vectors corresponding to the elements in $\mathcal{S}$. Then, by \cref{fact:local} there exists a local purity concentration protocol with error at most $O(\eps^{1/8})$ with rate 
    \begin{align*}
        \log \abs{Y_2} &\geq \log \abs{\mathcal{Y}}-\log \abs{\mathcal{S}} \\
        & \ge \log\abs{\mathcal{Y}}-I_{\max}^{\eps}(X:RB)_{\rho^{XRB}}
    \end{align*}

The net purity at the end of this step is then
    \begin{align*}
        & \log d_A -\log \abs{\mathcal{Y}}+\log \abs{\mathcal{Y}}-I_{\max}^{\eps}(X:RB)_{\rho^{XRB}}\\
         = &\log d_A-I_{\max}^{\eps}(X:RB)_{\rho^{XRB}}\;.
    \end{align*}

    \item Alice and Bob are now left with the state
    \[
    \sigma^{Y_1B}\coloneqq\frac{1}{\Tr[\Pi_{\mathcal{S}}\tau]}\left(\I^B\otimes\Pi_{ \mathcal{S}}^{Y}\cdot\tau^{YB} \right)^{Y_1}=\frac{1}{\Tr[\Pi_{\mathcal{S}}\tau]}\sum_{y\in \mathcal{S}}P_Y(y)\ketbra{y}^{Y_1}\otimes \rho_y^B \;,
    \]
    Alice then applies the bijection given by \cref{lem:slepianwolf} to create the state
    \[
    \sigma^{MNB}\coloneqq \sum_{m,n} P_{MN}(m,n) \ketbra{m,n}^{MN}\otimes \rho_{mn}^B
    \]
    where
    \begin{align*}
        \log N &\leq I_H^{\sqrt{\eps_0}}(Y_1:B)_{\sigma^{Y_1B}}+\log \eps_0 \\
        \intertext{and}
        MN &= \abs{Y_1} \\
        &= \abs{\mathcal{S}}
    \end{align*}
    Alice sends the systems $MN$ to Bob through the dephasing channel, which requires at most $\log\abs{\mathcal{S}}\leq I_{\max}^{\eps}(X:RB)_{\rho^{XRB}}$ number of bits. This quantity can be further bounded by $H_{\max}^{O(\eps^2)}(A)$, see Lemma~\ref{lem:Imax_Hmaxtilde}.
    
    \item Finally, after receiving the system $MN$, Bob applies the unitary $W^{MNB}$ given by \cref{corol:bobsprotocol} such that
    \[
\norm{\Tr_{MB}\left(W^{MNB}\cdot \sigma^{MNB}\right)-\ketbra{0}^N}_1 \leq 2\eps_0^{1/8}
\]
to distill $\log N$ amount of purity. 
\item The proof of \cref{corol:bobsprotocol} also tells us that the state on system $B$ after Bob applies the unitary $W^{MNN}$ is $2\eps_0^{1/8}$ away from 
\[
\frac{1}{\Tr[\Pi_{\mathcal{S}}\tau]}\sum_{y\in \mathcal{S}}P_Y(y) \rho_y^B \;.
\]
However, since $\mathcal{S}$ is a set of high probability under $P_Y$, this implies that
\[
\norm{\sum_yP_Y(y)\rho_y^B-\frac{1}{\Tr[\Pi_{\mathcal{S}}\tau]}\sum_{y\in \mathcal{S}}P_Y(y) \rho_y^B }_1\leq O(\eps_0^{1/16})\;,
\]
that is, in the end, Bob has a state which is not far from what he had at the beginning of the protocol, namely, 
\[
\rho^B=\sum_yP_Y(y)\rho_y^B\;.
\]
Thus, Bob can apply the local protocol on the system $B$ and recover $\log d_B-\widetilde{H}_{\max}^{\eps}(B)$ amount of purity with error $O(\eps_0^{1/8})$.
\item Summarizing, the total amount of purity distilled is 
\begin{align*}
    & \log d_Ad_B -I_{\max}^{\eps}(X:RB)_{\rho^{XRB}} - \widetilde{H}_{\max}^{\eps}(B)+I_H^{\sqrt{\eps_0}}(Y_1:B)_{\sigma^{Y_1B}}+\log \eps_0\;.
\end{align*}
Recall from our choice of POVM that
\[
I_H^{\sqrt{\eps_0}}(Y_1:B)_{\sigma^{Y_1B}} \geq I_H^{\eps_0/2}(X:B)_{\rho^{XB}}-O(1)+O(\log (1-\eps_0))\;.
\]
Therefore, the total amount of purity recovered by the protocol is at least
\[
\log d_Ad_B -I_{\max}^{\eps}(X:RB)_{\rho^{XRB}} - \widetilde{H}_{\max}^{\eps}(B)+I_H^{\eps_0/2}(X:B)_{\rho^{XB}}+\log \eps_0(1-\eps_0)-O(1)
\]
which by \cref{fact:smooth_maxinfo_equivalence} and \cref{lem:Imax_Hmaxtilde} is further lower bounded by:
\[
\log d_Ad_B -\widetilde{H}_{\max}^{\eps^2/48}(A)_{\rho^A} - \widetilde{H}_{\max}^{\eps}(B)+I_H^{\eps_0}(X:B)_{\rho^{XB}}+O(\log \eps)-O(1)\;.
\]
\item At each step the protocol made an additive error of at most $O(\eps_0^{1/4})=O(\eps^{1/8})$. Therefore the total error of the protocol is given by $O(\eps^{1/8})$.
\end{enumerate}

This finishes the proof of \cref{thm:purity_distillation}.
\end{proof}

\begin{remark}
Note that in the asymptotic i.i.d. limit, $I_{\max}^{\eps}(X:RB)$ becomes
\begin{align*}
    I(X:RB)_{\rho^{XRB}} &= H(RB)-H(RB|X) \\
    & = H(A)-\sum_x P_X(x)H(RB|x) \\
    & = H(A)
\end{align*}
where the last equality uses the fact that conditioned on each $x$, the state $\left.\rho^{RB}\right|_x$ is pure, since the POVM $\Tilde{\Lambda}$ is rank $1$.
\end{remark}

\section{Proof of main lemmas and corollaries} \label{sec:Main_proofs}

In this section we prove the main technical lemmas and corollaries required mentioned in \cref{sec:Technicaltools} for the purity distillation protocol of \cref{sec:MainProtocol}.

\subsection{Proof of \cref{lem:povm}} \label{sub:proof_lem:POVM}

The proof is subdivided into several parts.

Recall from the measurement compression theorem \cite{Winter_meascomp,Wilde_etal_measurementcomp, ChakrabortyPadakandlaSen_22} that we create a new POVM which itself is a convex combination of several compressed POVMs with a smaller number of outcomes. By construction, each compressed POVM has an element corresponding to the $0$-th outcome which corresponds to failure, which has probability at most $\eps$. Therefore, averaged over all choices of POVM, the total mass on the outcome corresponding to failure is at most $\eps$. Conditioning on success, the simulating POVM creates a distribution $P_{KL}$, where $k\in [K]$ is the index of the compressed POVM, and $\ell\in [L]$ corresponds to the outcome of the measurement using the compressed POVM. The simulating measurement then creates the following state
\[
\sigma^{KLRB}\coloneqq \sum_{k,\ell\in [K]\times [L]} P_{KL}(k,\ell)\ketbra{k,\ell}^{KL}\otimes \sigma_{k\ell}^{RB}\;,
\]
where
\[
\sigma_{k\ell}^{RB}\coloneqq \frac{1}{\Tr \Gamma_{k,\ell}^A(\rho^{RAB})}\Tr_A \Gamma_{k,\ell}^A(\rho^{RAB})\;.
\]
Recall also that there exists a map
\begin{align*}
f  :&[K]\times [L]\to \mathcal{X} \\
 & (k,\ell)\mapsto \tilde{x}
\end{align*}
such that, when this map is applied to $\sigma^{KLRB}$, the measurement compression theorem ensures that
\begin{align*}
    \norm{\sum_x P_X(x) \ketbra{x}^X\otimes \rho_x^{RB}-\sum_{\tilde{x}}Q_{\widetilde{X}}(\tilde{x})\ketbra{\tilde{x}}^{\widetilde{X}}\otimes \sigma_{\tilde{x}}^{RB}}_1 \leq O(\eps)\coloneqq \eps''\;.
\end{align*}
Tracing out the system $RB$, this implies that
\[
\norm{P_X-Q_{\widetilde{X}}}_1 \leq \eps''\;.
\]
Standard arguments then imply that there exists a subset $\good_X\subset \mathcal{X}$ such that
\begin{align*}
    &\Pr_{P_X}[\good_X] \geq 1-\sqrt{\eps''}\;, \\
    &Q_{\widetilde{X}}(x) \leq (1+\sqrt{\eps''})P_X(x),\quad \forall x\in \good_X\;.
\end{align*}

\begin{claim}\label{claim:1}
Let $\good_{KL}$ be the set of pairs $(k,\ell)$ such that $f(k,\ell)\in \good_X$. Then 
\[
\Pr_{P_{KL}}[\good_{KL}] \geq 1-2\sqrt{\eps''}\;.
\]
\end{claim}

\begin{proof}
Note that, for a fixed $x\in \good_X$,
\begin{align*}
    \Pr_{Q_{\widetilde{X}}}[x] &= \sum_{\substack{k,\ell\\ f(k,\ell)=x}} P_{KL}(k,\ell)
\end{align*}
Since $Q_{\widetilde{X}}$ and $P_X$ are close, this implies that
\begin{align*}
\Pr_{Q_{\widetilde{X}}}[\good_X]&=1-\sum_{x\notin\good_X}Q_{\widetilde{X}}(x)\\
&\geq 1-\eps''-\sum_{x\notin\good_X}P_{X}(x)\\
&\geq 1-\eps''-\sqrt{\eps''}\\
&\geq 1-2\sqrt{\eps''}\;.
\end{align*}
Therefore,
\begin{align*}
    \Pr_{P_{KL}}[\good_{KL}] & = \sum_{x\in \good_X}\sum_{\substack{k,\ell\\ f(k,\ell)=x}} P_{KL}(k,\ell)\\
    & = \sum_{x\in \good_X} \Pr_{Q_{\widetilde{X}}}[x] \\
    & = \Pr_{Q_{\widetilde{X}}}[\good_X] \\
    & \geq 1-2\sqrt{\eps''}
\end{align*}
\end{proof}
Now, for a fixed $k$, define the set $\bad_{L|k}$ to be the set of those $\ell$'s such that $(k,\ell)\notin \good_{KL}$ and $\good_{L|k}=L\setminus\bad_{L|k}$.
\begin{claim}\label{claim:2}
There exists a subset $\good_K\subset [K]$ such that
\[
\Pr_{P_K}[\good_K] \geq 1-2\sqrt[4]{\eps''}
\]
and, for all $k\in \good_K$,
\[
\Pr_{P_{L|k}}[\bad_{L|k}] \leq \sqrt[4]{\eps''}\;.
\]
\end{claim}
\begin{proof}
Define the indicator $\one_{k,\ell}$ to be $1$ when the pair $(k,\ell)\in \good_{KL}$. Also define  
\[
\eta_k\coloneqq \sum_{\ell}P_{L|k}(\ell|k)\one_{k,\ell}
\]
Then, we know from \cref{claim:1} that
\begin{align*}
    \sum_{k}P_K(k)\ \eta_k & = \sum_{k,\ell}P_{KL}(k,\ell)\one_{k,\ell} \\
    &=\Pr_{P_{KL}}[\good_{KL}]\\
    & \geq 1-2\sqrt{\eps''}\;.
\end{align*}
From the above relation, we have that
\[
\sum_{k}P_K(k)\ (1-\eta_k)\le 2\sqrt{\eps''}\;. 
\]
An application of Markov's inequality implies that
\[
\Pr_{P_K}\left[\brak{k~|~1-\eta_k\geq \sqrt[4]{\eps''}}\right] \leq 2\sqrt[4]{\eps''}\;,
\]
that is,
\[
\Pr_{P_K}\left[\brak{k~|~\eta_k\geq 1-\sqrt[4]{\eps''}}\right] \geq 1-2\sqrt[4]{\eps''}\;.
\]
Let us now define
\[
\good_K\coloneqq \brak{k~|~\eta_k\geq 1-\sqrt[4]{\eps''}}\;,
\]
and let us fix $k\in \good_K$. Then
\begin{align*}
    \Pr_{P_{L|k}}[\bad_{L|k}] &=\Pr_{P_{L|k}}\left[(k,\ell)\notin \good_{KL}\right] \\
    &= \sum_{\ell}P_{L|k}(\ell|k)\cdot (1-\one_{k,\ell}) \\
    & \leq \sqrt[4]{\eps''}\;.
\end{align*}
This concludes the proof.
\end{proof}
We will now define the new state
\[
\tsigma^{KLRB} \coloneqq \sum_{(k,\ell)\in \good_{KL}} \widetilde{P}_{KL}(k,\ell)\ketbra{k,\ell}^{KL}\otimes \sigma_{k\ell}^{RB}\;,
\]
where
\begin{align*}
    \widetilde{P}_{KL}(k,\ell):=\begin{cases}
    	\frac{1}{\Pr[\good_{KL}]}P_{KL}(k,\ell)& \text{if } (k,\ell)\in \good_{KL}\\
    	0&\text{otherwise}\;.
    \end{cases}
\end{align*}
By the previous discussion and standard manipulations, this implies that
\[
\norm{\sigma^{KLRB}-\tsigma^{KLRB}}\leq 4\sqrt{\eps''}\;.
\]
 We will now precisely quantify the state $\tsigma^{KLRB}$. Firstly, recall from the measurement compression theorem that, by choice of the simulating measurement, 
\[
\sigma^{KLRB}= c_0 \sum_{\substack{k\in [K']\\ \ell\in \mathcal{L}(k)}} \frac{1}{KL} \ketbra{k,\ell}^{KL}\otimes \sigma_{k\ell}^{RB}
\]
where $\mathcal{L}(k)$ is the subset of those indices $\ell\in [L]$ for the fixed  value of $k$ which correspond to valid outcomes of the compressed POVM. We also know that
\begin{align*}
    K' &= (1-\eps_0) K\;, \\
    \abs{\mathcal{L}(k)} &= (1-\eps_0) L\;, \quad \forall k\in [K'] \\
    c_0 &= \frac{1}{(1-\eps_0)^2}\;,
\end{align*}
where, recall that
\[
\eps_0\coloneqq O(\eps^{1/4})
\]
We further have the property that for all $k\in [K']$ The original POVM has outcomes in $\cX$ however; what is meant is that the bijection has been applied
\[
\frac{1}{L}\sum_{\ell\in \mathcal{L}(k)}\sigma_{k\ell}^{RB}\leq (1+\eps_0) \Tr_A[(\I^{RB}\otimes \Lambda^A)({\varphi_\rho}^{RAB})]
\]
where $\ket{\varphi}^{RAB}$ is some purification of $\rho^{AB}$ and $\Lambda^A$ is the original POVM. Next, we only retain those $k\in [K']$ that are in the set $\good_K$. Since we know that
\[
\Pr_{\textup{Unif}[K']}[\good_K] \geq 1-\eps_0
\]
we need only throw away a further $\eps_0$ fraction of the set $[K]$. Again, for every $k\in \good_K$, we throw away the set $\bad_{L|k}$.

Recall that from \cref{claim:2} we have that, for all $k\in \good_k$,
\[
\Pr_{\mathcal{L}(k)}[\bad_{L|k}]\leq \eps_0
\]
Therefore, we need only throw a further $\eps_0$ fraction of every set $\mathcal{L}(k)$. Thus
\begin{align*}
    \tsigma^{KLRB}= c_1 \sum_{\substack{k\in [(1-\eps_0)\cdot K')]\\ \ell\in \mathcal{L}(k)\setminus \bad_{L|k}}} \frac{1}{K\cdot L} \ketbra{k,\ell}^{KL}\otimes \sigma_{k\ell}^{RB}
\end{align*}
where 
\[
c_1= \frac{1}{(1-\eps_0)^4}
\]
To ease notation we define
\begin{align*}
    K'' & \coloneqq (1-\eps_0)\cdot K' \\
    \mathcal{L}'(k) & \coloneqq \mathcal{L}(k)\setminus \bad_{L|k}~~ \forall k\in [K'']
\end{align*}
We then have the properties that
\begin{align*}
    \frac{1}{L}\sum_{\ell\in \mathcal{L}'(k)}\sigma_{k\ell}^{RB} & \leq (1+\eps_0) \Tr_A[\I^{RB}\otimes \Lambda^A({\varphi_\rho}^{RAB})] \\
    \implies \frac{1}{L}\sum_{\ell\in \mathcal{L}'(k)}\sigma_{k\ell}^{B} & \leq \rho^B \\
    \intertext{and }
    \tsigma^X & \leq \frac{1+\eps_0}{1-\eps_0}\rho^X
\end{align*}
Thus, in keeping with our notation, we define
\begin{align*}
	\widetilde{P}_{KL}(k,\ell) &\coloneqq
	\begin{dcases}
		 \frac{c_1}{K\times L}, &\forall k\in [K''] \wedge \ell \in \mathcal{L}'(k) \\
		 0 &\textup{ otherwise}\;.
	\end{dcases}
\end{align*}
Before our next claim we need a couple of new definitions in order to define the conditional hypothesis testing mutual information and use it in the subsequent proofs.
\begin{definition}
Given the state $\tsigma^{KLB}$, for $k\in [K'']$ define the state $\tsigma^{LB}|k$ as
\[
\tsigma^{LB}|k \coloneqq \sum \limits_{\ell\in \mathcal{L}'(k)}\widetilde{P}_{L|k}(\ell|k)\ketbra{\ell}^L\otimes \sigma_{k,\ell}^B
\]
and the states $\tsigma^{L}|k\; \tsigma^{B}|k$ and $\left( \tsigma^{L} \otimes \tsigma^{B} \right)|K $ respectively as:
\begin{align*}
   & \tsigma^{L}|k \coloneqq \Tr_B(\tsigma^{LB}|k) = \sum \limits_{\ell\in \mathcal{L}'(k)}\widetilde{P}_{L|k}(\ell|k)\ketbra{\ell}^L \\
   & \tsigma^{B}|k \coloneqq \Tr_L(\tsigma^{LB}|k) = \sum \limits_{\ell\in \mathcal{L}'(k)}\widetilde{P}_{L|k} \sigma_{k,\ell}^B\\
   & \left(\tsigma^{L} \otimes \tsigma^{B} \right) |K \coloneqq \sum \limits_{k \in K''} \widetilde{P}_K(k) \ketbra{k}^K  \otimes \left\{ \left( \sum \limits_{\ell\in \mathcal{L}'(k)}\widetilde{P}_{L|k}(\ell|k)\ketbra{\ell}^L \right) \otimes \left( \sum \limits_{\ell\in \mathcal{L}'(k)}\widetilde{P}_{L|k}(\ell|k) \sigma_{k,\ell}^B \right) \right\}
\end{align*}
\end{definition}
Now using the above definitions, we define the two variants of conditional hypothesis testing mutual information as follows:
\begin{definition}
Given the state $\tsigma^{LB}|k$, we define
\begin{align*}
& I_H^{\eps_0}(L:B|k)_{\tsigma^{LB}|k}\coloneqq D_H^{\eps_0}(\tsigma^{LB}|k~||~\tsigma^L|k\otimes \tsigma^B|k)\\
& I_H^{\eps_0}(L:B|K)_{\tsigma^{KLB}} \coloneqq D_H^{\eps_0}(\tsigma^{KLB}~||~\left(\tsigma^L \otimes \tsigma^B \right)|K)
\end{align*}
\end{definition} 
\begin{remark}
We need the above definitions of conditional hypothesis testing mutual to come up with the derandomization or the expurgation argument to show that there exixts an index $k \in K''$ for which there exists a compressed POVM that preserves the classical correlations of the state $\rho^{XB}$. An exposition to the definition of conditional hypothesis testing mutual information can be found in \cite[Corollary~4]{sen2020unions}.
\end{remark}

 \begin{claim}\label{clm:previousClaim}Consider the quantity $D_H^{\eps_0}(\tsigma^{KLB}~||~\tsigma^{KL}\otimes \rho^B)$. Let $\Pi_{\textsc{OPT}}^{KLB}$ be the optimising operator for this quantity. It holds that
\[I_H^{\eps_0}(L:B|K)_{\tsigma} \geq I_H^{\eps_0/2}(X:B)_{\rho^{XB}}-O(1)+O(\log \eps)\;.\]
\end{claim}

\begin{proof}
Without loss of generality we can assume that
\[
\Pi_{\opt}= \sum_{(k,\ell)\in \textup{supp}(\widetilde{P}_{KL})}\ketbra{k}^K\otimes \ketbra{\ell}^L\otimes \Pi^B_{k,\ell}
\]
Then,
\begin{align*}
    2^{-I_H^{\eps_0}(L:B|K)_{\tsigma}}\leq & c_1^2 \cdot \Tr\Pi_{\opt}^{KLB}\left(\sum_{k\in [K'']}\frac{1}{K}\ketbra{k}^K\otimes \left(\sum_{\ell\in \mathcal{L}'(k)}\frac{1}{L}\ketbra{\ell}^L\right)\otimes \left(\sum_{\ell\in \mathcal{L}'(k)}\frac{1}{L}\sigma_{k,\ell}^{B}\right)\right) \\ 
    \leq & c_1^2 \cdot \Tr\Pi_{\opt}^{KLB}\left(\sum_{k\in [K'']}\frac{1}{K}\ketbra{k}^K\otimes \left(\sum_{\ell\in \mathcal{L}'(k)}\frac{1}{L}\ketbra{\ell}^L\right)\otimes (1+\eps_0)\rho^B \right) \\
    = & c_1^2 \cdot (1+\eps_0) \Tr \Pi_{\opt}^{KLB} \left(\sum\limits_{(k,\ell)\in \textup{supp}(\widetilde{P}_{KL})} \frac{1}{K\cdot L} \ketbra{k,\ell}^{KL}\otimes \rho^B\right) \\
    = & c_1^2 \cdot (1+\eps_0) 2^{-D_H^{\eps_0}(\tsigma^{KLB}~||~\tsigma^{KL}\otimes \rho^B)} \\
    \leq& c_1^2\cdot (1+\eps_0) 2^{-D_H^{\eps_0}(\tsigma^{XB}~||~\tsigma^X\otimes \rho^B)}
\end{align*}
where the last inequality is via the data processing inequality. \\
Next, suppose that $\lambda^{XB}$ is the optimising operator for 
\[
D_H^{\eps_0}(\rho^{XB}~||~\tsigma^X\otimes \rho^B)
\]
Further,
\[
\begin{aligned}
\Tr[\lambda^{XB}\tsigma^{XB}] \geq & \Tr[\lambda^{XB}\rho^{XB}]-\norm{\tsigma^{XB}-\rho^{XB}}_1 \\
\geq & 1-\eps_0
\end{aligned}
\]
This further implies that
\[
\begin{aligned}
2^{-I_H^{\eps_0}(L:B|K)_{\tsigma}} \leq  c_1^2(1+\eps_0)  2^{-D_H^{\eps_0/2}(\rho^{XB}~||~\tsigma^X\otimes \rho^B)}
\end{aligned}
\]
Finally, since we are also given that 
\[
\tsigma^X \leq \frac{1+\eps_0}{1-\eps_0}\rho^X
\]
we get
\[
\begin{aligned}
2^{-I_{H}^{\eps_0}(L:B|K)_{\tsigma}} \leq& c_1^2 \cdot (1+\eps_0)^2 \cdot \frac{1}{1-\eps_0} \cdot 2^{-D_H^{\eps_0/2}(\rho^{XB}~||~\rho^X\otimes \rho^B)} \\
= & c_1^2 \cdot (1+\eps_0)^2\cdot \frac{1}{1-\eps_0} \cdot 2^{-I_H^{\eps_0/2}(X:B)_{\rho^{XB}}}
\end{aligned}
\]
Therefore, it now follows from the definition of $c_1^2$ that
\[
I_H^{\eps_0}(L:B|K)_{\tsigma}\geq I^{\eps_0/2}_H(X:B)_{\rho^{XB}}-O(1)+O(\log (1-\eps_0)) 
\]
This concludes the proof.
\end{proof}

\begin{claim}
There exists a $k\in [K'']$ such that
\[
I_H^{\sqrt{\eps_0}}(L:B|k)\geq I_H^{\eps_0/2}(X:B)_{\rho^{XB}}-O(1)+O(\log (1-\eps_0))
\]
\end{claim}
\begin{proof}
First, let $\Gamma_{\opt}$ be the optimizer for $I_H^{\eps_0}(L:B|K)$. Since 
\[
\Tr[\Gamma_{\opt}\tsigma^{KLB}]\geq 1-\eps_0
\]
this implies that
\begin{align*}
    \sum_{k\in [K'']}\wP(k)\Tr[\Gamma(k)^{LB}\tsigma^{LB}|k]\geq 1-\eps_0
\end{align*}
where $\Gamma(k)^{LB}$ is simply the operator $\Gamma_{\opt}$ for a fixed $k$. This is well defined since $\Gamma_{\opt}$ is cq. Then, Markov's inequality implies that, there exists a subset $\mathcal{S}\subset [K'']$ of probability at least $1-\sqrt{\eps_0}$ such that for all $k\in \mathcal{S}$, 
\[
\Tr[\Gamma(k)^{LB}\tsigma^{LB}|k] \geq 1-\sqrt{\eps_0}
\]
Then,
\begin{align*}
    2^{-I_{H}^{\eps_0}(L:B|K)_{\tsigma}} & = \Tr \left[ \Gamma_{\opt}^{KLB}\left(\sum_{k\in [K'']}\wP(k)\ketbra{k}^K\otimes \left(\sum_{\ell\in \mathcal{L}'(k)}\wP(\ell|k)\ketbra{\ell}^L\right)\otimes \left(\sum_{\ell\in \mathcal{L}'(k)}\wP(\ell|k)\tsigma_{k,\ell}^{B}\right)\right) \right]\\
    & = \sum \limits_{k\in [K'']} \wP(k) \Tr \left[ \Gamma_{k}^{LB} \tsigma^L|k \otimes \tsigma^B|k \right]\\
    & \geq \sum\limits_{k\in \mathcal{S}} \wP(k) \Tr \left[ \Gamma_{k}^{LB} \tsigma^L|k \otimes \tsigma^B|k \right]\\
    \intertext{then by definition of $I_H^{\sqrt{\eps_0}}(L:B|k)_{\tsigma^{LB}|k}$ we can lower bound the above expression by}
    & \geq \sum\limits_{k\in \mathcal{S}} \wP(k) 2^{-I_H^{\sqrt{\eps_0}}(L:B|k)_{\tsigma^{LB}|k}} \\
    & = \left(\sum_{k\in \mathcal{S}}\wP(k)\right)\sum\limits_{k\in \mathcal{S}} \frac{\wP(k)}{\sum_{k\in \mathcal{S}}\wP(k)} 2^{-I_H^{\sqrt{\eps_0}}(L:B|k)_{\tsigma^{LB}|k}} \\
   & \geq (1-\sqrt{\eps_0})\sum\limits_{k\in\mathcal{S}} \widetilde{Q}(k)  2^{-I_H^{\sqrt{\eps_0}}(L:B|k)_{\tsigma^{LB}|k}}
\end{align*}
  where we define
  \begin{align*}
  \widetilde{Q}(k):=\begin{cases}
  \frac{\wP(k)}{\sum_{k\in \mathcal{S}}\wP(k)} &\forall k\in \mathcal{S} \\
  0 & \text{otherwise}\;,	
  \end{cases}
  \end{align*}
  and the last line follows by the fact that $\mathcal{S}$ has probability at least $1-\sqrt{\eps_0}$ with respect to the distribution $\wP(k)$. This implies that there exists a $k\in [K'']$ such that
\[
I_H^{\eps_0}(L:B|K) \leq I_H^{\sqrt{\eps_0}}(L:B|k)_{\tsigma^{LB}|k}+\log \frac{1}{1-\sqrt{\eps_0}}
\]
Finally by \cref{clm:previousClaim}, we see that
\[
I_H^{\eps_0/2}(X:B)_{\rho^{XB}}-O(1)+O(\log (1-\eps_0)) \leq I_H^{\sqrt{\eps_0}}(L:B|k)_{\tsigma^{LB}|k}+\log\frac{1}{1-\sqrt{\eps_0}}
\]
This concludes the proof.
\end{proof}

Note that the compressed sub POVM which corresponds to this $k$ was made by conditioning on success and then removing the outcomes in the set $\bad_{L|k}$. To create a rank-one POVM from this sub POVM we will first use the spectral decomposition of the POVM element corresponding to failure and append these weighted rank-one projectors to the original sub POVM. We will also append the projectors corresponding to those outcomes in the set $\bad_{L|k}$ to the sub POVM. Note that all the rank-one elements that we appended to the sub POVM to complete it, account for at most $2\eps_0$ amount of probability mass. This gives the POVM $\widetilde{\Lambda}$ that we will use in the protocol.

\begin{remark}
    Note that in the protocol, the quantity $I_H^{\sqrt{\eps_0}}(L:B|k)_{\tsigma^{LB}|k}$ is referred to as $I^{\sqrt{\eps_0}}_H(Y_1:B)_{\sigma^{Y_1B}}$. See \cref{sec:MainProtocol} for details.
\end{remark}

\subsection{Proof of \cref{lem:slepianwolf}}\label{subsec:proof_lem:slepianwolf}

In this subsection we prove \cref{lem:slepianwolf}. It contains an important claim for a deterministic binning strategy in \cref{claim:permutation_hash}, as opposed to the random binning strategy used to obtain inner bound for classical message compression with quantum side information \cite{Devetak_Winter_Slepianwolf, Renes_Renner,Wilde_etal_measurementcomp}. We further show that, even in the case of deterministic binning, the encoding is similar to choosing a $2$-universal hash function randomly. In our case such a property is true by choosing a random permutation over the classical indices.

We will think of $\mathcal{X}$ as a $2$ dimensional array stored in the memory of a computer, where each cell in the array is addressed by a tuple $(m,n)$ where we call $m$ the \emph{block index} or simply index, and $n$ the \emph{intra block} index. Let $\sigma$ be permutation picked randomly from the set of all permutations on the elements of the set $\mathcal{X}$. The action of $\sigma$ is to reassign the elements in the set $\mathcal{X}$ to a possibly different location in the array. Let $f_{\sigma}$ be the function that maps each $x\in\mathcal{X}$ to the corresponding block number $m$, with respect to a fixed permutation $\sigma$ :
\begin{align*}
f_{\sigma}:&\ \mathcal{X}\to [M] \\
          &\ x \mapsto m\;,\textup{ such that } \sigma(x)=(m,n)\;.
\end{align*}

\begin{claim}\label{claim:permutation_hash}
It holds over the choice of the random permutation $\sigma$ that, for any $x\neq x'$,
\[
\Pr[f_{\sigma}(x)=f_{\sigma}(x')]=\frac{N-1}{\abs{\mathcal{X}}-1}\;.
\]
\end{claim}

\begin{proof}
We evaluate the collision probability as follows :
\begin{align*}
    \Pr[f_{\sigma}(x)=f_{\sigma}(x')] & = \sum_{\substack{m\in [M]\\ n,n'\in [N]\\ n\neq n'}}\Pr[\sigma(x')=(m,n')~|~\sigma(x)=(m,n)]\cdot \Pr[\sigma(x)=(m,n)] \\
    & = \sum_{\substack{m\in [M]\\  n\in [N]}}\sum_{\substack{n'\neq n\\ n'\in [N]}}\Pr[\sigma(x')=(m,n')~|~\sigma(x)=(m,n)]\cdot \frac{1}{\abs{\mathcal{X}}} \\
    & = \sum_{\substack{m\in [M]\\ n\in [N]}} \frac{N-1}{\abs{\mathcal{X}}-1}\cdot \frac{1}{\abs{\mathcal{X}}} \\
    & = \frac{N-1}{\abs{\mathcal{X}}-1}
\end{align*}
\end{proof}

\subsubsection{The protocol}

For the purposes of the proof we will imagine a communication protocol between Alice and Bob, which attempts to accomplish the following task.

Alice receives a symbol $x$ from the distribution $P_X$. She sends to Bob some encoding of this symbol over a classical noiseless channel. Bob has access to the state $\rho_x^{B}$ as side information. Conditioned on Alice's message, Bob will perform a measurement on the system $B$ to produce a guess for the symbol $x$. We require that Alice send as few bits as possible while ensuring that the symbol $x$ and Bob's guess are equal with high probability on average. This task is formally known as classical data compression with quantum side information. It is well known \cite{Wilde_etal_measurementcomp} that a protocol for this task implies the existence of a map $\mathcal{D}$ which satisfies condition $3$ in the statement of \cref{lem:slepianwolf}. The protocol will proceed as follows:
\begin{enumerate}
    \item Let $\sigma$ be a randomly chosen permutation on the set $\mathcal{X}$; $\sigma$ is stored as a public coin, and is accessible to both parties Alice and Bob.
    \item Upon receiving the symbol $x$, Alice sends Bob the index $f_{\sigma}(x)$.
    \item Bob then performs a sequential measurement, based on the received index. Our choice of measurement operators will ensure that Bob's decoding will succeed with high probability, as long as the number of indices to distinguish satisfies
    \[
    \log N \leq I_{H}^{\eps}(X:B)_{\rho^{XB}}+\log \eps\;.
    \]
\end{enumerate}
Details follow.

\subsubsection{Bob's decoding}

We wish to analyze the probability of a decoding error. We first define the set $\mathcal{A}(f_{\sigma},m)$ as the intersection between the pre-image of $m$ under $f_{\sigma}$ and the support of $P'_X$, where $P'_X$ is the sub-distribution that is obtained by removing those points from the support of $P_X$ that have the smallest probabilities that add to at most $\eps$:
\[
\mathcal{A}(f_{\sigma},m)\coloneqq \brak{x~|~f_{\sigma}(x)=m, x\in \textup{supp}(P')}
\]
Let us denote the elements in $\mathcal{A}(f_{\sigma},m)$ as $\brak{a_1^{m},a_2^{m},\ldots, a_{\abs{\mathcal{A}}}^{m}}$. Each $a_i^{m} \in \cX$.

Now consider the operator $\Pi_{\textsc{opt}}$ from the definition of $I_{H}^{\eps}(X:B)_{\rho}$ and the associated operators $\Pi_x$, which are derived from the classical-quantum form of $\Pi_{\textsc{opt}}$. Upon receiving the index $m$, Bob sequentially measures his system $B$ with the operators $\Pi_{a_i^{m}}$. The probability of incurring in a decoding error can be evaluated as follows.
 
First of all, we will not work with the distribution $P_X$ but with the sub-distribution $P'_X$. To see that this only incurs an extra $\eps$ error, note that
\begin{align*}
    &\Pr[\textup{decoding error}] \\
    &= \sum_x P_X(x)\cdot \Pr[\textup{decoding error}~|~x] \\
    & = \sum_{x\in \textup{supp}(P')}P_X(x)\cdot \Pr[\textup{decoding error}~|~x] + \sum_{x\notin \textup{supp}(P')}P_X(X)(x)\cdot \Pr[\textup{decoding error}~|~x] \\
    & \leq \sum_{x}P'_X(x)\cdot \Pr[\textup{decoding error}~|~x] + \eps\;.
\end{align*}
Suppose that when the encoded $x$ is sent, the corresponding representation of this symbol in the set $\mathcal{A}(f_{\sigma},m)$ is $a_{\ell}^{m}$, where $\ell\in [N]$.  Then, conditioned on Alice having received $x$, the probability of incorrect decoding is given by
\begin{align*}
\Pr[\textup{decoding error}~|~x] &= 1-\Tr\left[\Pi_{a_{\ell}^{m}}(\I-\Pi_{a_{\ell-1}^{m}})\ldots (\I-\Pi_{a_1^{m}})\cdot \rho_{a_{\ell}^{m}}\right]\\
&= \Tr\left[\rho_{a_{\ell}^{m}}\right]-\Tr\left[\Pi_{a_{\ell}^{m}}(\I-\Pi_{a_{\ell-1}^{m}})\ldots (\I-\Pi_{a_1^{m}})\cdot \rho_{a_{\ell}^{m}}\right]\;. \\
\intertext{Using Sen's non-commutative union bound \ref{fact:noncommutative_union_bound}, the above expression can by bounded by}
\leq &\sqrt{\Tr\left[(\I-\Pi_{a_{\ell}^{m}}) \rho_{a_{\ell}^{m}}\right]+\sum\limits_{i=1}^{\ell-1}\Tr\left[\Pi_{a_i^{m}} \rho_{a_{\ell}^{m}}\right]}\;.
\end{align*}

Now notice that, the sets $\mathcal{A}(f_{\sigma},m)$ form a disjoint cover of the set $\textsc{supp}\left(P'_X\right)$ over the indices $m$. Thus, taking an average over the elements of the set $\bigcup\limits_{m}\mathcal{A}(f_{\sigma},m)$ is the same as taking an average over the set $\textsc{supp}\left(P'_X\right)$. Using this observation along with the concavity of the square root, we see that the average error probability over choices of $x$ is at most
\begin{align*}
    &\sqrt{\sum_{m\in [M]}\sum_{a_{\ell}^{m}\in \mathcal{A}(f_{\sigma},m)}P'_X(a_{\ell}^{m})\left(\Tr\left[(\I-\Pi_{a_{\ell}^{m}}) \rho_{a_{\ell}^{m}}\right]+\sum\limits_{i=1}^{\ell-1}\Tr\left[\Pi_{a_i^{m}} \rho_{a_{\ell}^{m}}\right]\right)} \\
    = &\sqrt{\sum_xP'_X(x)\Tr\left[(\I-\Pi_{x}) \rho_{x}\right]+\sum_{m\in [M]}\sum_{a_{\ell}^{m}\in \mathcal{A}(f_{\sigma},m)}P'_X(a_{\ell}^{m})\sum\limits_{i=1}^{\ell-1}\Tr\left[\Pi_{a_i^{m}} \rho_{a_{\ell}^{m}}\right]}\;.
\end{align*}
The first term inside the square root is at most $\eps$, by the property of $\Pi_{\textsc{opt}}$ that
\begin{align*}
\Tr\left[ \Pi_{\textsc{opt}}\sum_xP_X(x)\ketbra{x}^X\otimes\rho_x\right] &\geq 1-\eps\;, \\
\intertext{which implies that}
\Tr\left[\sum_x P_X(x)\ketbra{x}^X\otimes \Pi_x\rho_x\right]&\geq 1-\eps\;,
\end{align*}
and the same holds for the pruned distribution $P'_X$ of course.

To analyze the second term inside the square root, consider the following :
\begin{align*}
    &\sum_{m\in [M]}\sum_{a_{\ell}^{m}\in \mathcal{A}(f_{\sigma},m)}P'_X(a_{\ell}^{m})\sum\limits_{i=1}^{\ell-1}\Tr\left[\Pi_{a_i^{m}} \rho_{a_{\ell}^{m}}\right] \\
    \leq &\sum_{m\in [M]}\sum_{a_{\ell}^{m}\in \mathcal{A}(f_{\sigma},m)}P'_X(a_{\ell}^{m})\sum\limits_{i\neq \ell}\Tr\left[\Pi_{a_i^{m}} \rho_{a_{\ell}^{m}}\right] \\
    =& \sum_x P'_X(x) \sum_{\substack{x'\neq x \\ x'\in \textsc{supp}(P')}}\one_{\brak{f_{\sigma}(x')= f_{\sigma}(x)}} \Tr\left[\Pi_{x'}\rho_x\right]
\end{align*}
where $\one_{\brak{f_{\sigma}(x')= f_{\sigma}(x)}}$ is the indicator for when $f_{\sigma}(x')= f_{\sigma}(x)$. We will now take an expectation over the choices of the random permutation $\sigma$. Note that the above term is inside a square root, so to do this we use the concavity of square root:
\begin{align*}
   &\E_{\sigma}\left[ \sum_x P'_X(x) \sum_{x'\neq x}\one_{\brak{f_{\sigma}(x')= f_{\sigma}(x)}} \Tr\left[\Pi_{x'}\rho_x\right]\right] \\
   = & \sum_x P'_X(x) \sum_{\substack{x'\neq x \\ x'\in \textsc{supp}(P')}}\E_{\sigma}\left[\one_{\brak{f_{\sigma}(x')= f_{\sigma}(x)}}\right] \Tr\left[\Pi_{x'}\rho_x\right] \\
   =& \sum_x P'_X(x) \sum_{\substack{x'\neq x\\ x'\in \textsc{supp}(P')}}\Pr[f_{\sigma}(x')= f_{\sigma}(x)] \Tr\left[\Pi_{x'}\rho_x\right] \\
   \leq & \frac{N-1}{\abs{\mathcal{X}}-1}\sum_x P'_X(x) \sum_{\substack{x'\neq x\\ x'\in \textsc{supp}(P')}}\Tr\left[\Pi_{x'}\rho_x\right] \\
   \leq & \frac{N}{\abs{\mathcal{X}}}\sum_x P'_X(x) \sum_{x'\in \textsc{supp}(P')}\Tr\left[\Pi_{x'}\rho_x\right] \\
   \intertext{To bound this term, we multiply and divide by $P'_X(x')$ inside the second summation. This shows us that the above expression is equal to}
   =& \frac{N}{\abs{\mathcal{X}}}\sum_x P'_X(x) \sum_{x'\in \textsc{supp}(P')}\frac{P'_X(x')}{P'_X(x')}\Tr\left[\Pi_{x'}\rho_x\right] \\
   \leq & \frac{N}{\abs{\mathcal{X}}}\cdot 2^{(\hmax')^{\eps}(X)} \sum_xP'_X(x)\Tr\left[\sum_{x'\in \textsc{supp}(P')}\left(P'_X(x')\Pi_{x'}\right)\rho_x\right]
\end{align*}
where we have used \cref{def:Hmax} to upper bound each $\frac{1}{P'_X(x')}$ term by $2^{(\hmax')^{\eps}(X)}$, which in turn is upper bounded by $\frac{\abs{\mathcal{X}}}{\eps}$ using \cref{prop:Hmax_inequalities}. We will now switch back to the distribution $P_X$ by adding the terms corresponding to the $x$'s which not in the support of $P'_X$. This implies that the above expression can be upper bounded by 
\begin{align*}
    \leq & \frac{N}{\abs{\mathcal{X}}}\cdot \frac{\abs{\mathcal{X}}}{\eps} \sum_x P_X(x)\Tr\left[\sum_{x'}\left(P_X(x')\Pi_{x'}\right)\rho_x\right] \\
    =& \frac{N}{\eps}\Tr\left[\left(\sum_{x'}P_X(x')\Pi_{x'}\right)\left(\sum_x P_X(x)\rho_x\right)\right] \\
    =& \frac{N}{\eps}\Tr\left[\left(\sum_{x'}P_X(x')\ketbra{x'}^X\otimes\Pi_{x'}^B\right) \left\{ \I^X \otimes\left(\sum_x P_X(x) \rho_x^B\right) \right\} \right] \\
    =& \frac{N}{\eps}\Tr\left[\left(\sum_{x'}\ketbra{x'}^X\otimes\Pi_{x'}^B\right) \left\{ \left( \sum_{x''}P_X(x'')\ketbra{x''}^X\right)\otimes\left(\sum_x P_X(x) \rho_x^B\right) \right\} \right] \\
    =& \frac{N}{\eps}\Tr\left[ \Pi_{\textsc{opt}}~ (\rho^X\otimes \rho^B) \right] \\
    = & 2^{\log N-I_{H}^{\eps}(X:B)_{\rho}+\log \frac{1}{\eps}}\;.
\end{align*}
Thus, this shows that as long as
\[
\log N \leq I_{H}^{\eps}(X:B)_{\rho}+2\log \eps\;,
\]
the average decoding error over choices of $x$ and the permutation $\sigma$ is at most $\sqrt{2\eps}+\eps$.

To finish the proof, consider the left polar decomposition of the operator
\[
\Pi_{a_{\ell}^{m}}(\I-\Pi_{a_{\ell-1}^{m}})\ldots (\I-\Pi_{a_1^{m}}) = U_{a_{\ell}^{m}}\sqrt{\Theta_{a_{\ell}^{m}}}
\]
where $\Theta_{a_{\ell}^{m}}$ is some positive operator. It is not hard to see that, for each $m\in [M]$, the operators $\Theta_{a^m_{\ell}}$ obey the operator inequality
\[
\sum_{\ell} \Theta_{a_{\ell}^m}^B\leq \I^B.
\]
Thus we extend these operators to a POVM by simply assigning appropriate operators to those $x$'s in $f_{\sigma}^{-1}(m)$ which are not in $\textup{supp}(P'_X)$. If no such $x$'s exist, we simply assign one operator which completes the POVM to the abort outcome, denoted by $\bot$. Now suppose that $a_{\ell}^m$ corresponds to some symbol $x$. Then, by standard manipulations we have that (see also \cref{fact:Gentle_operator_lemma})
\[
\norm{\rho_x-\sqrt{\Theta_{x}}\rho_{x}\sqrt{\Theta_{x}}}_1 \leq 2\sqrt{\Tr\left[(\I-\Theta_{x})\rho_x\right]}\;.
\]
It is now easy to see that the following bounds hold
\begin{align*}
    &\sum_xP_X(x)\norm{\rho_x-\sqrt{\Theta_{x}}\rho_{x}\sqrt{\Theta_{x}}}_1 \\
    \leq & \sum_xP'_X(x)\norm{\rho_x-\sqrt{\Theta_{x}}\rho_{x}\sqrt{\Theta_{x}}}_1+2\eps \\
    \leq & 2\sum_x P'_X(x)\sqrt{\Tr\left[(\I-\Theta_{x})\rho_x\right]}+2\eps \\
    \leq & 2 \sqrt{\sum_xP'_X(x)\Tr\left[(\I-\Theta_{x})\rho_x\right]}+2\eps\;, \\
    \intertext{which, by our previous computations, implies that the above expression can be upper bounded by}
    \leq & 2\sqrt{\sqrt{2\eps}+\eps}+2\eps = O(\eps^{1/4})= \eps_0\;.
\end{align*}
One can now derandomise the argument to conclude the existence of a permutation for which the above conditions hold. Relabelling the symbols $x$ and the POVM elements appropriately shows us that the third condition holds. This concludes the proof of \cref{lem:slepianwolf}.

\subsection{Proof of \cref{corol:bobsprotocol}} \label{subsec:proof_corol:bobsprotocol}

Given the state 
\[
 \sum_{m,n} P_{MN}(m,n) \ketbra{m,n}^{MN}\otimes \rho_{mn}^B\;,
\]
consider, for some fixed $m$, the state
\[
\sigma_{N|m}^{NB}\coloneqq \sum_{n}P_{N|m}(n|m)\ketbra{n}^N\otimes \rho^{B}_{mn}\;.
\]
For each such conditioned state, we can add one more outcome, with corresponding null probability, to account for the outcome $\bot$. Since this is weighted with zero probability, it does not change the average state. Now for each $m \in M$ consider the following unitary operators
\[
W^{NB}(m)\coloneqq \sum_{n', n\in [N]\cup\brak{\bot}} \ket{n'}\bra{n}^N\otimes \Gamma_{n,n'}^{B}(m)\;,
\]
where
\[
\Gamma_{n,0}^B(m)\coloneqq \sqrt{\Theta_n(m)}\;,
\]
and the other operators can always be chosen appropriately to satisfy unitarity. Next, define the unitary
\[
W^{MNB}\coloneqq \sum_{m}\ketbra{m}^M\otimes W^{NB}(m)\;,
\]
and consider the following
\begin{align*}
    &W^{MNB}\cdot \sum_{m}P_M(m)\ketbra{m}^M\otimes \sigma_{N|m}^{NB} \\=& \sum_{m,n}P_{MN}(m,n)\ketbra{m}^M\otimes \ketbra{0}^N\otimes \sqrt{\Theta_n(m)}\rho_{mn}^B\sqrt{\Theta_n(m)} +\tau^{MNB}_{\textup{error}}\;,
\end{align*}
where $\tau_{\textup{error}}$ is some Hermitian matrix that results from the terms with $n'\neq 0$ of $W^{NB}(m)$. Therefore,
\[
\bra{0}^N \tau_{\textup{error}}\ket{0}^N=0
\]
by definition. Now, from the result of \cref{lem:slepianwolf} we know that
\begin{align*}
    &\left\lVert\sum_{m,n}P_{MN}(m,n)\ketbra{m}^M\otimes \ketbra{0}^N\otimes  \rho_{nm}^{B}\right.\\ &\left. -\sum_{m,n}P_{MN}(m,n)\ketbra{m}^M\otimes \ketbra{0}^N\otimes \sqrt{\Theta_n(m)}\rho_{mn}^B\sqrt{\Theta_n(m)} \right\rVert_1 \\
    \leq & \sum_{m,n}P_{MN}(m,n)\norm{\rho_{mn}^B-\sqrt{\Theta_n(m)}\rho_{mn}^B\sqrt{\Theta_n(m)}}_1 \\
    \leq &\;\eps'.
\end{align*}
In particular, these observations imply that
\[
\Tr\left[\sum_{m,n}P_{MN}(m,n)\ketbra{m}^M\otimes \ketbra{0}^N\otimes \sqrt{\Theta_n(m)}\rho_{mn}^B\sqrt{\Theta_n(m)} \right]\ge 1-\eps'\;,
\]
that is,
\[
\begin{aligned}
\bra{0}\Tr_{MB}\left[W^{MNB}\cdot \sigma^{MNB}\right]\ket{0} &= \sum_{m,n}P_{MN}(m,n)\Tr\left[ \Theta_n(m)\ \rho_{mn}^B\right] \\
& \geq 1-\eps'\;.
\end{aligned}
\]
Thus, we now see that
\begin{align*}
    &\norm{\Tr_{BM}\left(W^{BMN}\cdot \sigma^{BMN}\right)-\ketbra{0}^N}_1 \\
     & \leq2\sqrt{1-\bra{0}\Tr_{BM}\left(W^{BMN}\cdot \sigma^{BMN}\right)\ket{0}} \\
    & \leq 2\sqrt{\eps'}\;.
\end{align*}
Next note that the matrix $\Tr_{MN}[\tau_{\textup{error}}^{MNB}]$, being by construction a diagonal block of a positive semi-definite matrix, is positive semi-definite. To see this more explicitly, notice that for a fixed $m$
\begin{align*}
    W^{NB}(m)\cdot \sigma^{NB}_{N|m} &= \sum_n P_{N|m}(n|m)\ketbra{0}^N\otimes \sqrt{\Theta_n(m)}\cdot \rho_{mn}^B \\
    +& \sum_n\sum_{\substack{n'',n' \\ \neg \brak{n'=0\wedge n''=0}}} P_{N|m}(n,m)\ket{n''}\bra{n'}^N\otimes \Gamma_{n,n'}^B(m)\rho_{mn}^B\Gamma_{n,n''}^B(m)
\end{align*} this is true for all $m$, tracing out the $M$ and $N$ systems with respect to the matrix $\tau_{\textup{error}}^{MNB}$ leaves the matrix
\[
\sum_{\substack{m,n, n' \\ n'\neq 0}}P_{MN}(n,m)\Gamma^B_{n,n'}(m)\cdot \rho_{mn}^B
\]
which is positive semi-definite. This implies that
\[
\norm{\tau_{\textup{error}}^B}_1=\Tr[\tau_{\textup{error}}^B] \leq \eps'\;.
\]
Therefore,
\begin{align*}
    \norm{\Tr_{NM}\left(W^{BMN}\cdot \sigma^{BMN}\right)-\sum_{m,n}P_{MN}(m,n)\sqrt{\Theta_{n}(m)}\cdot \rho_{mn}^B}_1 &= \norm{\tau^{B}}_1 \\
    & \leq \eps'\;.
\end{align*}
Finally, since
\[
\begin{aligned}
&\norm{\sum_{m,n}P_{MN}(m,n)\rho^B_{m,n}-\sum_{m,n}P_{MN}(m,n)\sqrt{\Theta_{n}(m)}\cdot \rho_{mn}^B}_1 \leq \eps'\;,
\end{aligned}
\]
we conclude that
\[
\norm{\Tr_{NM}\left(W^{BMN}\cdot \sigma^{BMN}\right)-\sum_{m,n}P_{MN}(m,n) \rho_{mn}^B}_1 \leq 2\eps'\;.
\]
This concludes the proof.

\section{Conclusion} \label{sec:Conclusion}
In this paper, we studied the achievable rate for the distillation of pure qubit states from a given mixed state of a single and bipartite system, when only a single copy of the state is available. For a single-party system, we proved that an achievable one-shot rate is given by
\begin{align*}
	\approx \log d_A - \widetilde{H}^\eps_{\max}(A)\;,
\end{align*}
whereas in the bipartite scenario, a one-shot achievable rate is
\begin{align*}
	\approx \log d_Ad_B-\widetilde{H}_{\max}^{\eps^2/48}(A)+\log d_B-\widetilde{H}_{\max}^\eps(B) + \max I_H^{\sqrt{\eps_0}}(X;B)\;,
\end{align*}
where the quantum mutual information-like term represents the increase with respect to the purely local protocol enabled by the use of one-way classical communication. Both these rates approach the given rates in \cite{Devetak_purity} in the asymptotic i.i.d. limit. We leave the generalization to a multi-party setting and the connection to the one-shot distillable common randomness \cite{Devetak_Winter} for future work.

\section*{Acknowledgments}

SC and AN would like to thank Pranab Sen and Rahul Jain for several helpful discussions, suggestions and several bits of advice which helped guide this project to fruition.

SC would like to acknowledge support from the National Research Foundation, including under NRF RF Award No. NRF-NRFF2013-13 and NRF2021-QEP2-02-P05 and the Prime Minister’s Office, Singapore and the Ministry of Education, Singapore, under the Research Centres of Excellence program. SC would also like to acknowledge support from the Google Late PhD Fellowship grant. AN and FB acknowledge support from MEXT Quantum Leap Flagship Program (MEXT QLEAP) Grant No. JPMXS0120319794. FB acknowledge support also from MEXT-JSPS Grant-in-Aid for Transformative Research Areas (A) ``Extreme Universe'', No. 21H05183, and from JSPS KAKENHI Grants No. 19H04066 and No. 20K03746.

\bibliographystyle{alpha}
\bibliography{references}

\appendix

\section{Preliminaries} \label{sec:Preliminaries}

\subsection{Notation}\label{sec:Notation}

In what follows, all sets ($\cX,\dots$) are finite and Hilbert spaces ($\cH,\dots$) are finite-dimensional. Besides the standard notation and nomenclature widely used in quantum information theory (see, e.g., Ref.~\cite{wilde_2013}), we use the term `substate' to mean a positive semidefinite matrixe with trace less than unity. The set of all substates on a Hilbert space $\cH$ is denoted as $\cS_{\leq}$ and the set of all states (i.e., density matrices) is denoted as $\cS_{=}$. We denote by $\I$ the identity operator.

For brevity, given two matrices $A$ and $B$ we use the notation $A \cdot B$ to denote $A B A^\dagger$. The symbol $\circ$ denotes sequential composition.
For any operator $M$ , we use $\lVert M \rVert_p := [\Tr{(M^\dagger M)}^{p/2}]^{1/p}$ to denote the Schatten $p$-norm of $M$: in particular, $\norm{M}_\infty$ is called as the operator norm and given by its maximal singular value. The trace can then be defined as $\norm{M}_1 := \Tr[\sqrt{M^\dagger M}]$. 
A positive operator-valued measure (POVM) $\Lambda$ on Hilbert space $\cH$ is a family of positive semidefinite matrices (called as the POVM elements) $ \{ \Lambda_x \}_{ \{ x \in \cX \} }$, such that and $\sum_{x \in \cX} \Lambda_x = \I$. 

We use the notation $\Pr$ to denote the probability of an event. We use plain capital letters to denote a random variable and $P_X(\cdot)$ to denote the probability mass function (pmf) of the random variable $X$ over an alphabet or the sample space $\cX$. The notation $\textup{supp}(P_X)$ is used to denote the support of the pmf $P_X$, i.e., the subset of $x\in\cX$ for which $P_X(x) > 0$.  

Concerning information-theoretic quantities, we adopt the following conventions:

\begin{itemize}
	\item  $H(\cdot)$ to denotes the entropy, either Shannon's or von Neumann's, depending on the context; the subscript (if any) denotes the state with respect to which it is calculated. The quantity $I( \cdot \; : \; \cdot)$ will be used to denote the mutual information. All the logarithms are taken in base $2$;
	
	\item We use the generalized fidelity~\cite{Tomamichel_thesis} as a distance measure for smoothing purpose, unless sated otherwise. This is defined as
	\[
	F(\rho, \sigma) := \lVert \sqrt{\rho}\sqrt{\sigma} \rVert_1 + \sqrt{(1-\Tr[\rho])(1-\Tr[\sigma])}
	\]
	for any pair of positive semidefinite operators with $\Tr \leq 1$, $\rho, \; \sigma$. Observe that when at least one of the above operators $\rho, \sigma $ is normalized to have unit trace, then
	\[
	F(\rho, \sigma) = \norm{\sqrt{\rho}\sqrt{\sigma}}_1\;,
	\]
	which corresponds to the standard definition for fidelity.
\end{itemize}

\subsection{Useful mathematical facts}

We recall two important facts that will be used in the protocols:

\begin{fact}[{Gentle Operator Lemma \cite{Winter_gentle_measurement,GentleMeasurement}}] \label{fact:Gentle_operator_lemma} 
	For a given state $\rho\in \cS_{=}(\cH^A)$ and a positive
	operator $\Lambda^A \leq I^A$ such that
	$\Tr (\Lambda \rho) \geq 1 - \eps$ for a given $\eps >0 $, it holds that:
	\[
	\norm{\rho^A - \sqrt{\Lambda} \rho \sqrt{\Lambda}}_1 \leq 2\sqrt{ \eps} 
	\]
\end{fact}

\begin{fact}[{\cite[Lemma~2]{Sen_noncommutative}}]\label{fact:noncommutative_union_bound} 
	Let $\rho \in \cS_{\leq}$. Let $\Pi_1, \Pi_2, \ldots, \Pi_k$
	be projectors. Let $\Pi_i' := I - \Pi_i$ be the projector onto the
	subspace orthogonal to the support of $\Pi_i$. Then,
	\[
	\Tr[\Pi_k \Pi_{k-1} \ldots \Pi_1 \rho ] \geq \Tr[\rho]- 2 \sqrt{  \sum \limits_{i=1}^k \Tr[\Pi_i' \rho]}
	\]
\end{fact}

We now state a recent one-shot measurement compression with quantum side information result. This will be pivotal for deriving a one-shot achievable rate for the purity distillation protocol to be discussed in \cref{sec:MainProtocol}. We state it as the following fact:

\begin{fact}[{One-shot measurement compression with quantum side information \cite[Proposition~4.7]{ChakrabortyPadakandlaSen_22}}] \label{fact:meas_comp_QSI}
	Given the shared quantum state $\rho^{AB}$, where the receiver Bob possesses the $B$ system, one-shot measurement compression with quantum side information can be achieved with a rate of classical communication
	\begin{align*}
		R_X &> {}^1 I_{\max}^{\eps}(X:RB) - I^{\epsilon_0/2}_H(X:B) 
		+ O(\log \epsilon^{-1})+1 
	\end{align*}
	where $\epsilon_0\coloneqq \eps^{1/10}$. All the entropic quantities in the above equation are computed with 
	respect to 
	$
	\sum\limits_{x\in \cX}
	\ketbra{x}^{X} \otimes \Tr_A[(\Lambda_{x}^A \otimes \I^{BR})\ \rho^{A B R}]\;.
	$
\end{fact}

The quantity ${}^1 I_{\max}^{\eps}(X:RB)$ is the smoothed max-information defined in \cref{def:Imax1}.


\section{Useful entropic quantities used in the protocols and their properties} \label{sec:entropies}

As we make extensive use of several entropic quantities, we devote this section to introduce the relevant quantities for our protocols. We also discuss the relationship among them, and show that in the asymptotic i.i.d. limit they indeed converge to their analogous Shannon (or von Neumann, in the quantum case) versions.


\subsection{Standard entropic quantities} \label{subsec:usual_entropies}

\begin{definition}[Shannon and von Neumann entropy]
	For an $\cX$-valued random variable $X$ with $X \sim P_X$, its Shannon entropy is defined as:
	\[
	H(X):=-\sum_{x \in \cX} P_X(x) \log P_X(x)\;.
	\]
	Analogously, for a density matrix $\rho^A$, its von Neumann entropy is defined as:
	\[
	H(A)_\rho:= -\Tr[\rho \log \rho]\;,
	\]
	where, if $\rho^A = \sum_{i \in A} \lambda_i \ketbra{i}^A$ then $\log \rho^A:=  \sum \limits_{i \in \textup{ support }(A)} \log (\lambda_i) \ketbra{i}^A$. Note that the von Neumann entropy coincides with the Shannon entropy of the eigenvalues of $\rho^A$. When no confusion arises, we shall use the term entropy to denote the Shannon entropy or the von Neumann entropy, depending on the case at hand. 
\end{definition}

\begin{definition}[quantum relative entropy]
	For any two positive semidefinite operators $\rho, \; \sigma$, the quantum relative entropy between is defined as:
	\[
	D(\rho||\sigma):= \begin{cases}
		\Tr[\rho \log \rho- \rho \log \sigma]&\text{if }	\operatorname{supp}\rho\subseteq \operatorname{supp}\sigma\;,\\
		+\infty&\text{otherwise.}
	\end{cases}
	\]
\end{definition}

\begin{definition}[quantum mutual information]
	For a given bipartite state $\rho^{AB} \in \cS_{=}(\cH^A \otimes \cH^B)$ the quantum mutual information between systems $A$ and $B$ is defined as:
	\[
	I(A:B)_\rho:=D(\rho^{AB}||\rho^A \otimes \rho^B) = H(A)_\rho+ H(B)_\rho - H(AB)_\rho\;.
	\]
\end{definition}

Various one-shot analogues of $D(\rho||\sigma)$ can be given, but in this work we will mostly focus on $D_{max}$ and $D_H^\eps$ relative entropies~\cite{Datta,Buscemi_Datta}, each with its own operational meaning. In the same vein, one-shot analogues of the quantum mutual information can be formulated. We refer the reader to \cite{QAEP, Tomamichel_thesis, TomamichelHayashi, TomamichelTan} for a detailed exposition to this topic.

\begin{definition}[max-relative entropy~\cite{Datta}]
	For a pair of quantum states $\rho,\; \sigma$ the max-relative entropy is defined as:
	\[
	D_{\max}(\rho||\sigma) := \begin{cases}
		\min{ \{ \lambda: \rho \leq 2^\lambda \sigma \} }
		\equiv \log \norm{\sigma^{-1/2} \rho \sigma^{-1/2}}_\infty\;,&\text{if }	\operatorname{supp}\rho\subseteq \operatorname{supp}\sigma\\
		+\infty\;,&\text{otherwise}\;.
	\end{cases}
	\]
\end{definition}

\begin{definition}[max-entropy~\cite{Tomamichel_thesis}]\label{def:max_entropy}
	Given a distribution $P_X$, the max-entropy of that distribution denoted by $\hmax(X)$, is defined as
	\[
	\hmax(X)\coloneqq 2 \log \sum_x \sqrt{P_X(x)}\;.
	\]
	Similarly, for a quantum state $\rho^A \in \cS_{=}(\cH^A)$, the max-entropy and the $\eps$-smooth max-entropy of $\rho^A$ for any $\eps \in [0,1]$ is defined as:
	\[
	\begin{aligned}
		&\hmax(A)_\rho := 2 \log \norm{\sqrt{\rho}}_1\\
		&\hmax^\eps(A)_\rho := \inf_{\tilde{\rho}^A \in \cB^\eps(\rho^A)} 2 \log \norm{\sqrt{\tilde{\rho}}}_1\;.
	\end{aligned}
	\]
	The max-entropy as defined above coincides with the R\'{e}nyi entropy of order $1/2$ \cite{Renner_minmax}.
\end{definition}

\begin{definition}[conditional min-entropy \cite{Renner_thesis, Datta}]\label{def:Hmin}
	The quantum conditional min-entropy for a bipartite state $\rho^{AB}$ is defined as:
	\begin{align*}
		H_{\min}(A|B)_\rho :&= -\min_{\sigma^B \in \cS_{=}(\cH^B)}D_{\max}(\rho^{AB}\| I^A\otimes \sigma^B)\\
		&= - \log \min_{\sigma^B \in \cS_{\le}(\cH^B)} \{
		\Tr[\sigma^B] : \rho^{AB} \leq I^A \otimes \sigma^B \}\;.
	\end{align*}
	The $\eps$-smoothed quantum conditional $\min$-entropy for a state $\rho^{AB}$ and a value $\eps \in [0,1]$ is defined as:
	\[
	H_{\min}^\eps (A|B)_\rho := \sup_{\tilde{\rho}^{AB} \in \cB^\eps (\rho^{AB})}H_{\min}(A|B)_{\tilde{\rho}}\;.
	\]	
\end{definition}

\begin{definition}[hypothesis-testing relative entropy and mutual information \cite{Buscemi_Datta}]\label{def:I_H}
	The hypothesis testing relative entropy between states $\rho$ and $\sigma$, for a given $\eps\in[0,1]$ is defined as:
	\[
	D_H^\eps (\rho||\sigma):= \sup_{\substack{{0 \leq \Lambda \leq I:}\\{\Tr[\Lambda \rho] \geq 1-\eps}}} -\log \Tr[\Lambda \sigma]\;.
	\]
	Further, the hypothesis testing mutual information is defined as:
	\[
	I_H^\eps(A:B)_\rho := D_H^{\eps}(\rho^{AB} || \rho^A \otimes \rho^B)\;.
	\]
\end{definition}
\begin{remark}
	Without loss of generality we can assume that the optimizing operator $\Lambda$ for any smooth hypothesis testing mutual information quantity which is defined with respect to a classical-quantum state is also classical-quantum in the sense that it is classical on the classical system.
\end{remark}

\begin{definition}[max-mutual information] \label{def:Imax1}
	For a given a bipartite state $\rho^{AB}$ and a value $\eps\in[0,1]$, the max-information and the $\eps$-smoothed max-information are respectively defined as:
	\[
	\begin{aligned}
		&I_{\max}(A : B)_{\rho} := D_{\max} (\rho^{AB} || \rho^A \otimes \rho^B),\\
		&I_{\max}^{\eps}(A : B)_{\rho} := \min_{
			\tilde\rho^{AB} \in \cB^\eps(\rho^{AB})}
		{I_{\max} (A : B)_{\tilde\rho}}\;. 
	\end{aligned}
	\]
\end{definition}


\subsection{New modified one-shot entropic quantities} 

In this subsection we introduce alternative versions of the one-shot entropic quantities mentioned above in \cref{subsec:usual_entropies}. We shall use two new versions of the smoothed max-entropy defined as follows:

\begin{definition}[\cite{SenNotes21}, smoothed support max-entropy] \label{def:smooth_supp_maxentropy}
	Let $\eps \in [0,1]$. Given a distribution $P_X$, let $\mathcal{S}_{\textsc{bad}} \subset \mathcal{X}$ consist of those symbols in the support of $P_X$, which correspond to the smallest probabilities under $P_X$ that add up to $\epsilon$, and define the sub-distribution $P'_X$ as
	\[
	P'_X(x) = \begin{cases}
		P_X(x)\;,&	\forall x\in \mathcal{X}\setminus \mathcal{S}_{\textsc{bad}}\;, \\
		0\;,& \text{otherwise.}
	\end{cases}
	\]
	Then, the $\eps$-smoothed support max-entropy of $X$ is defined as
	\[
	\widetilde{H}_{\max}^{\eps}(X)\coloneqq \log \abs{\operatorname{supp}(P'_X)}\;.
	\]
	The quantum smoothed support max-entropy is defined analogously. Let $\rho^A \in \cS_{=}(\cH^A)$ and ${\rho'}^A \leq \rho^A$ be a substate obtained by zeroing out the smallest eigenvalues of $\rho^A$ that add up to at most $\eps$. Then:
	\[
	\widetilde{H}_{\max}^\eps(A)_\rho:= \log \abs{\operatorname{supp}(\rho')}\;.
	\]
\end{definition}

\begin{definition}[\cite{NemaSen_decoupling}, smoothed norm max-entropy]\label{def:Hmax}
	Under the notations of \ref{def:smooth_supp_maxentropy}, we define another variant of the max-entropy, which we shall refer to as the smoothed norm max-entropy, as:
	\[
	(\hmax')^\eps(X):= \log \norm{(P'_X)^{-1}}_\infty
	\]
	We also define the quantum smoothed norm max-entropy analogously. Then:
	\[
	(\hmax')^\eps(A)_\rho:= \log \norm{(\rho')^{-1}}_\infty
	\]
	where inverses are understood as pseudo-inverses, that is, inverses defined only on the support and zero elsewhere.
\end{definition}

We now define a mutual information-like quantity and its $\eps$-smoothed version as follows:

\begin{definition}[modified max-information~\cite{Ciganovic_2014}]\label{def:Imax2}
	For $\rho^{AB} \in \cS_{\leq}(\cH^{AB})$ and $\eps \geq 0$, the modified max-information and its $\eps$-smoothed modified max-information, for $\eps>0$, are defined as:
	\[
	\begin{aligned}
		&{}^{1}{I_{\max}(A : B)_\rho }:= \min_{\sigma^B \in \cS_{=}(\cH^B)}D_{\max} (\rho^{AB}||\rho^A \otimes \sigma^B),\\
		&{}^{1}{I_{\max}^\eps (A : B)_\rho} := \min_{
			\rho' \in \cB^\eps(\rho)}
		{}^{1}{I_{\max} (A : B)_{\rho'}} 
	\end{aligned}
	\]
\end{definition}


\subsection{Useful relations between various entropic quantities}

In this subsection we mention the useful relations between aforementioned entropic quantities, some of which are already available in the literature, and others which are new to this work. These relations serve the following two goals for us:
\begin{enumerate}
	\item Relate ${}^1 I_{\max}^\eps(X;B)_{\rho^{XB}}$ with $\widetilde{H}_{\max}^\eps(A)_{\rho^A}$ circumventing one of the shortcomings of the dud protocol in \cref{Protocol:dud}, which gives a term equal to $\widetilde{H}_{\max}^\eps(X)_{P_X}$ in the purity distillation rate for Alice. Hence, this leads to the application of a compressed POVM obtained from the one-shot measurement compression theorem of \cite{ChakrabortyPadakandlaSen_22}, resulting in the term ${}^1 I_{\max}^\eps(X;B)_{\rho^{XB}}$.
	\item Relate smoothed support and smoothed norm max-entropies in \cref{def:smooth_supp_maxentropy} and \cref{def:Hmax}, respectively, with each other, and with the standard smoothed max-entropy, which enables us to ensure that both these modified entropies converge to their Shannon analogue in the asymptotic i.i.d. limit. 
\end{enumerate}

\begin{fact}[{ \cite[Lemma~B.15]{Berta_RevShannon}}]\label{fact:Berta}
	Given a bipartite quantum state $\rho^{AB}$, 
	\[
	{}^{1}I_{\max}^{\eps}(A:B)_{\rho}\leq H_{\max}^{\eps^2/48}(A)_{\rho}-H_{\min}^{\eps^2/48}(A|B)_{\rho}-2\log \frac{\eps^2}{24}\;.
	\]
\end{fact}

The smoothed-max mutual information in \cref{def:Imax1} and \ref{def:Imax2} above are equivalent up to smoothing parameters:

\begin{fact}[{\cite[Theorem~3]{Ciganovic_2014}}]  \label{fact:smooth_maxinfo_equivalence}
	
	Let $\rho^{AB}\in \cS_{=}(\cH^{A} \otimes \cH^B)$, $\eps'>0$ and $\eps\in(0,1/4)$. Then there exists a real valued function $g$ of $\eps$, such that $g(\eps)=O(\log \frac{12}{\eps^2})$
	and the following equivalence holds:
	\begin{align*}
		{}^{1}I_{\max}^{\eps + 2 \sqrt{\eps}+\eps'}(A:B)_\rho & \leq I_{\max}^{\eps + 2 \sqrt{\eps}+\eps'}(A:B)_\rho \\
		& \leq {}^{1}I_{\max}^{\eps'}(A:B)_\rho+g(\eps)\;. 
	\end{align*}
\end{fact}

We now relate these different definitions of the smoothed max-entropies using the following proposition: 

\begin{proposition}
	\label{prop:Hmax_inequalities}
	From \cref{def:max_entropy} it is straightforward to see that: 
	\[
	\hmax^\eps(A)_\rho \leq \widetilde{H}_{\max}^\eps(A)_\rho \leq (\hmax')^\eps(A)_{\rho} \leq \frac{d_A}{\eps}.
	\]
\end{proposition}

\begin{proof}
We work with the notation in \cref{def:Hmax}. We begin by arranging the eigenvalues, say $\{\lambda_i\}_{i=1}^{d_A}$ of $\rho^A$ in ascending order, that is, $\lambda_1 \leq \lambda_2 \leq \ldots \leq \lambda_{d_A}$ and corresponding eigenvectors as $\ket{v_1}, \ket{v_2}, \dots, \ket{v_{d_A}}$. We subsume zero eigenvalues in the above ordering, if $\rho$ has rank less than $d_A$. Define the subset $\mathcal{S}_{\textsc{bad}}$ as the set of the first $k$ eigenvectors $\brak{\ket{v_1}, \ket{v_2} \ldots, \ket{v_k}}$ of $\rho$  which correspond to the eigenvalues $\brak{\lambda_1, \lambda_2,\ldots, \lambda_k}$ such that
\[
\sum_{1}^{k}\lambda_k \leq \eps
\]
and 
\[
 \sum_{i=1}^k \lambda_i + \lambda_{k+1}> \eps.
\]
Also define the projector $\Pi_{\mathcal{S}}$ as the projector onto the span of the vectors in $\mathcal{S}_{\textsc{bad}}$, i.e,:
\[
\Tr[\Pi_{\mathcal{S}} \rho]=\sum_{i=1}^k \lambda_i \leq \eps
\]
Thus, $\rho':= \sum \limits_{i=k+1}^{d_A} \lambda_i \ketbra{v_i}$.
By Definitions~\ref{def:max_entropy} and~\ref{def:smooth_supp_maxentropy} we have:
\begin{align*}
\hmax^\eps(A)_\rho &\leq   2 \log (\Tr[\sqrt{\rho'}])\\
&= 2 \log (\sum_{i \geq k+1} \sqrt{\lambda_i})\\
&\le 2 \log \left( \sqrt{|\textup{supp}(\rho')|} \times \sqrt{\sum_{i \geq k+1} \lambda_i} \right)\\
&\leq 2 \log (\sqrt{|\textup{supp}(\rho')|})\\
&= \widetilde{H}_{\max}^\eps(A)_\rho\;,
\end{align*}
where the second inequality comes from an application of the Cauchy--Schwarz inequality.
Also by the above ordering on the eigenvalues of $\rho^A$, we have $(\hmax')^\eps(A)_\rho=\log \frac{1}{\lambda_{k+1}}$. Then, since $1 \geq \sum_{i \geq k+1} \lambda_i\ge \abs{\textup{supp}(\rho')} \lambda_{k+1}$, we have that
\begin{align*}
\frac{1}{\lambda_{k+1}} \geq \abs{\textup{supp}(\rho')}\;, 
\end{align*}
or equivalently, $ (\hmax')^\eps (A)_\rho \geq \widetilde{H}_{\max}^\eps(A)_\rho$.
We also have $\lambda_{k+1} \geq \lambda_i, \; \forall\; i \leq k+1$.
This implies:
\begin{align*}
\sum \limits_{i=1}^{k+1} \lambda_{k+1} \geq \sum \limits _{i=1}^{k+1}\lambda_i \geq \eps\;,
\end{align*}
and therefore $(k+1)\lambda_{k+1} \geq \eps$, that is, $\lambda_{k+1} \geq \frac{\eps}{k+1} \geq \frac{\eps}{d_A}$.
Hence, \[\widetilde{H}_{\max}^\eps(A)_{\rho} \leq (\hmax')^\eps(A)_\rho \leq \frac{|d_A|}{\eps}\;.\]
This finishes the proof.
\end{proof}


\subsubsection{Relation between ${}^{1}I_{\max}^\eps(X:B)$ and $\widetilde{H}_{\max}^\eps(A)$} \label{subsec:Imax_upperbound}

We now prove the following important lemma that help us relate the smoothed max-information \cref{def:Imax1} obtained from the measurement compression theorem reported in~\cite{ChakrabortyPadakandlaSen_22} with the smoothed support max-entropy \cref{def:smooth_supp_maxentropy}. We begin by proving the following claim, which, even though very simple, is essential in proving the main lemma of this subsection, \cref{lem:Imax_Hmaxtilde}: 

\begin{claim}\label{claim:positivity}
	For a given classical quantum quantum state 
	\[
	\rho^{XR}= \sum_x P_X(x)\ketbra{x}^X\otimes \rho_x^R\;,
	\]
	it holds that
	\[
	H_{\min}^{\eps}(R|X)\geq 0\;.
	\]
\end{claim}
\begin{proof}
	The claim can be proved simply by using the definitions of conditional and smoothed conditional min-entropies for a classical-quantum state mentioned in \cref{def:Hmin} and the property that $H_{\min}(R)_{\rho_x} \geq 0$, for every $x$.\\
\end{proof}

\begin{lemma} \label{lem:Imax_Hmaxtilde}
	Given a state $\rho^A$ on system $A$, let $R$ be a reference purifying $A$, so that $\Tr_R[\ketbra{\psi}^{RA}]=\rho^A$. Given a POVM $\brak{\Lambda_x^A}$ on $A$, consider the resulting cq-state 
	\[
	\sigma^{XR} \coloneqq \sum_x P_X(x) \ketbra{x}^X\otimes \rho_x^R\;,
	\]
	where $\rho_x^R$ is the state of the reference conditional on the outcome $x$, that is,
	\[
	\rho_x^R \coloneqq \frac{1}{\Tr[\Lambda_x \rho]}\left(\sqrt{\rho}\Lambda_x^T \sqrt{\rho} \right)^R\;.
	\]
	Then, for any 
	$0 < \eps \leq \frac{1}{4}$ and $0 < \eps^{\prime} \leq \min\{ \frac{\eps}{3}, \frac{1}{4} \}$, it holds that
	\begin{align}
		&{}^{2}I_{\max}^{\eps}(X:R)_{\sigma^{XR}} \leq  \widetilde{H}_{\max}^{O(\eps^2)}(A)_{\rho}-O(\log \eps)\; \label{eq:smooth_Imax2_Hmax}  \\ 
		&I_{\max}^{2 \eps}(X:R)_{\sigma^{XR}}  \leq \widetilde{H}_{\max}^{O(\eps^2)}(A)_{\rho}-O ( \log \eps ) + O (\log \frac{\eps'^2}{12}) \label{eq:smooth_Imax1_Hmax}
	\end{align}
\end{lemma}
\begin{proof}
	From \cref{fact:Berta} we know that
	\[
	{}^{2}I_{\max}^{\eps}(X:R)_{\sigma}\leq H_{\max}^{\eps^2/48}(R)_{\sigma}-H_{\min}^{\eps^2/48}(R|X)_{\sigma}-2\log \frac{\eps^2}{24}
	\]
	Further, from \cref{claim:positivity}  we have that for any classical quantum state and any value of $\eps \in [0,1]$,
	\[
	H_{\min}^{\eps^2/48}(R|X)_{\sigma} \geq 0
	\]
	Note also that from the definition of the state $\sigma^{XR}$, 
	\begin{align*}
		\sigma^{R} & = \sum_x P_X(x)\rho_x^R = \sum_x \sqrt{\rho} \Lambda_x^T \sqrt{\rho}= \rho^R\;,
	\end{align*}
	which has the same eigenvalues of $\rho^A$. Therefore,
	\[
	H_{\max}^{\eps^2/48}(R)_{\sigma}= H_{\max}^{\eps^2/48}(A)_{\rho^A}\;.
	\]
	Finally by \cref{prop:Hmax_inequalities} we can upper bound $\hmax^{\eps^2/48}(A)_\rho$ by $\widetilde{H}_{\max}^{\eps^2/48}(A)_\rho$,
	and hence
	\[
	{}^{2}I_{\max}^{\eps}(X:R)_{\sigma^{XR}}\leq \widetilde{H}_{\max}^{\eps^2/48}(A)_{\rho}-2\log \frac{\eps^2}{24}
	\]
	This concludes the proof of \cref{eq:smooth_Imax2_Hmax}. The proof for \cref{eq:smooth_Imax1_Hmax} is a straightforward application of \cref{fact:smooth_maxinfo_equivalence} and the observation that $I_{\max}^\eps(A:B)_\rho$ and ${}^1I_{\max}^\eps(A:B)_\rho$ are both non-increasing in $\eps$.
\end{proof}

\begin{remark}
	In our main purity distillation protocol the register $R$ above is taken to be the joint register $RB$ and the above lemma can be directly applied.
\end{remark}


\subsection{Quantum asymptotic equipartition property (QAEP) and asymptotic i.i.d.  limits}
In this subsection we provide the asymptotic i.i.d.  limits of above defined one-shot smoothed entropic quantities and identify that they indeed converge to their Shannon analogues in the limit. The property that ensures that smoothed one-shot entropies converge to Shannon entropies in the i.i.d.  limit is often called as the \emph{asymptotic equipartition property} or AEP.

\begin{fact}[{Asymptotic i.i.d.  limit of $\hmax^\eps(A)$ \cite[Corollary~6.6 and Corollary~6.7]{Tomamichel_thesis}}] \label{fact:Hmaxi.i.d.}
	Given a state $\rho^A$,
	\[
	\lim\limits_{\eps\to 0}\lim\limits_{n\to \infty} \frac{1}{n} \hmax^{\eps}(A^{\otimes n}) = H(A)_\rho\;,
	\]
	where $H(A)_\rho$ denotes the von Neumann entropy of $\rho^A$.
\end{fact}

\begin{fact}[{Asymptotic i.i.d.  limit of $(\hmax')^\eps(A)$ \cite[Proposition~2]{NemaSen_decoupling}}] \label{fact:Hmax'_i.i.d.}
	Given a density matrix $\rho^A$,
	\[
	\lim\limits_{\eps\to 0}\lim\limits_{n\to \infty} \frac{1}{n} (\hmax')^{\eps}(\rho^{\otimes n}) = H(A)_\rho\;.
	\]
\end{fact}

\begin{proposition}[{Asymptotic i.i.d.  limit of $(\widetilde{H}_{\max})^\eps(A)$}]\label{prop:Hmax_tilde_i.i.d.} 
	Given a density matrix $\rho^A$,
	\[
	\lim \limits_{\eps\to 0}\lim \limits_{n\to \infty} \frac{1}{n} (\widetilde{H}_{\max})^{\eps}(\rho^{\otimes n}) = H(A)_\rho\;.
	\]
\end{proposition}

\begin{proof}
	The proof is a direct application of \cref{prop:Hmax_inequalities} and the sandwich property of limits.
\end{proof}

Finally, our rate expression for purity distillation protocol also includes $I_H^\eps$ and $I_{\max}^\eps$. The following fact provides the asymptotic i.i.d.  limits of these quantities. Note that these one-shot mutual information quantities are derived from $D_H^\eps$ and $D_{\max}^\eps$, respectively, see \cref{def:I_H} and \cref{def:Imax1}. We refer the reader to \cite{QAEP,TomamichelHayashi, TomamichelTan, KeLi, Ciganovic_2014} for the proof. 
\begin{fact}\label{fact:QAEP}
	Given the quantum states $\rho^{A}$ and $\sigma^A$, the following hold:
	\begin{align*}
		&\lim \limits_{\eps\to 0} \lim \limits_{n\to \infty}\frac{1}{n} D_H^{\eps}(\rho^{\otimes n}\|\sigma^{\otimes n}) = D(\rho\|\sigma) \\
		 &\lim \limits_{\eps\to 0} \lim \limits_{n\to \infty}\frac{1}{n} D_{\max}^{\eps}(\rho^{\otimes n}\|\sigma^{\otimes n}) = D(\rho\|\sigma)\;,
	\end{align*}
	where $D(\rho\|\sigma)$ is the quantum relative entropy.
\end{fact}

\begin{fact}\label{fact:Imax_I_H_i.i.d.}
	For any bipartite quantum states $\rho^{AB}$,
	\begin{align*}
		&\lim \limits_{\eps\to 0} \lim \limits_{n\to \infty}\frac{1}{n} I_H^{\eps}(A^n:B^n)_{\rho^{\otimes n}} = I(A:B)_\rho\;, \\
		& \lim \limits_{\eps\to 0} \lim \limits_{n\to \infty}\frac{1}{n} I_{\max}^{\eps}(A^n:B^n)_{\rho^{\otimes n}} = I(A:B)_\rho\;.
	\end{align*}
\end{fact}

\end{document}